\documentclass[pra,superscriptaddress,twocolumn,showpacs]{revtex4-2}
\usepackage{epsfig,amsmath,amssymb,mathrsfs,mathtools,amscd}
\usepackage[linkcolor=red]{hyperref}

\usepackage{cleveref}
\usepackage{soul}
\usepackage{enumerate}
\usepackage{amsmath} 
\usepackage{amsthm}
\usepackage{verbatim} 
\usepackage{mathtools} 
\usepackage{float} 
\usepackage{tikz}
\usepackage{stackengine}

\usepackage[normalem]{ulem}
\newcommand{\stkout}[1]{\ifmmode\text{\sout{\ensuremath{#1}}}\else\sout{#1}\fi}

\theoremstyle{plain}
\newtheorem{lemma}{Lemma}

\newtheorem{proposition}{Proposition}

\def\Tr{\operatorname{Tr}}

\def\>{\rangle} \def\<{\langle}

\def\T+{\mathsf{T}_+}

\newcommand{\bra}[1]{\langle#1|}
\newcommand{\ket}[1]{|#1\rangle}
\newcommand{\braket}[2]{\langle#1|#2\rangle}
\newcommand{\ketbra}[2]{{\ket{#1}\bra{#2}}}
\newcommand{\Bra}[1]{\langle \! \langle#1|}
\newcommand{\Ket}[1]{|#1\rangle \! \rangle}

\newcommand{\KetBra}[2]{{\Ket{#1}\Bra{#2}}}
\newcommand{\hilb}[1]{\mathcal{#1}}

\newcommand{\CNOT}{\rm  CNOT}

\begin{document}
	
	\title{Storage and retrieval of  two unknown unitary channels}

	\author{Michal \surname{Sedl\'ak} }
	\affiliation{RCQI, Institute of Physics, Slovak Academy of Sciences, D\'ubravsk\'a cesta 9, 84511 Bratislava, Slovakia}
	\affiliation{Faculty of Informatics,~Masaryk University,~Botanick\'a 68a,~60200 Brno,~Czech Republic}
	\author{ Robert \surname{St\'arek}}
	\affiliation{Department of Optics,  Faculty of Science, Palack\'{y} University, 17. listopadu 1192/12, 77900 Olomouc, Czech Republic}
	\author{Nikola  \surname{Horov\'a}}
	\affiliation{Department of Optics,  Faculty of Science, Palack\'{y} University, 17. listopadu 1192/12, 77900 Olomouc, Czech Republic}
	\author{Michal \surname{Mi\v cuda} }
	\affiliation{Department of Optics,  Faculty of Science, Palack\'{y} University, 17. listopadu 1192/12, 77900 Olomouc, Czech Republic}
	\author{Jaromir \surname{Fiur\'a\v sek}}
	\affiliation{Department of Optics,  Faculty of Science, Palack\'{y} University, 17. listopadu 1192/12, 77900 Olomouc, Czech Republic}

\author{Alessandro 
  \surname{Bisio}} \email[]{alessandro.bisio@unipv.it}
\affiliation{Dipartimento di Fisica dell'Universit\`a di Pavia, via
  Bassi 6, 27100 Pavia} \affiliation{Istituto Nazionale di Fisica
  Nucleare, Gruppo IV, via Bassi 6, 27100 Pavia}

	\date{\today}

	\begin{abstract}
		
		We address the fundamental task of converting
		$n$ uses of an unknown unitary transformation
		into a quantum state (i.e., storage) and later
		retrieval of the transformation. Specifically, we
		consider the case where the unknown unitary
		is selected with equal prior probability from
		two options. First, we prove that the optimal
		storage strategy involves 
		the sequential application of the $n$ uses of the unknown unitary,
		and it produces the
		optimal state for discrimination between the
		two possible unitaries. 
		Next, we show that incoherent
		"measure-and-prepare" retrieval achieves the
		maximum fidelity between the retrieved operation
		and the original (qubit) unitary. We then
		identify the retrieval strategy that maximizes
		the probability of successfully and perfectly
		retrieving the unknown transformation. In the
		regime in which the fidelity between the
		two possible unitaries is large the probability of success
		scales as
		$ P_{succ} = 1 - \mathcal{O}(n^{-2} ) $,
        which is a quadratic improvement  
        with respect to the case
		in which the unitaries are drawn from the entire
		unitary group
		$U(d)$ with uniform prior 
		probability. 
		Finally, we present an optical
		experiment for this approach and assess the
		storage and retrieval quality using quantum
		tomography of states and processes. The results
		are discussed in relation to non-optimal
		measure-and-prepare strategy, highlighting the
		advantages of our protocol.

	\end{abstract}

	
	\maketitle
	
	\section{Introduction}
	
	The non-orthogonality of quantum states is a
	distinctive feature of quantum theory and the
	fundamental reason why quantum information cannot
	be cloned \cite{wootters1982single,buvzek1996quantum,werner1998optimal}, gathered without disturbance \cite{PhysRevLett.68.557,Busch2009}, or have its "logical value" inverted \cite{PhysRevA.59.4238,PhysRevA.60.R2626}. A
	simple system of two quantum states that cannot be
	perfectly discriminated provides significant
	insights into many aspects of quantum theory
	\cite{Fuchs2002}. This remains true when considering
	transformations as carriers of quantum
	information. A transformation 
	between possibly different types of quantum systems
	is, in fact, a more
	general concept than that of a quantum state, which can be  
	considered
	as
	a special case 
	of transformation, whose input system is trivial.
	As a consequence, there are 
	more ways in which we can
	process transformations 
	and this leads to a  new paradigm for quantum computation and information processing in which transformations are processed by  quantum circuits with open slots \cite{selinger2004towards,gutoski2007toward,PhysRevLett.101.060401,PhysRevA.80.022339,Actaphysicanetwork} or, more generally, by  higher order maps
	which may also exhibit indefinite causal 
	order\cite{oreshkov2012quantum,araujo2017quantum,bisio2019theoretical}. Many tasks within this paradigm were considered  like discrimination\cite{acin2001statistical,duan2007entanglement,harrow2010adaptive,PhysRevLett.125.080505}, cloning  \cite{chiribella2008optimal,chiribella2015universal}, information-disturbance tradeoff \cite{bisio2010information,10.1088/1402-4896/ad7912}, inversion \cite{bisio2011minimal,PhysRevLett.123.210502,yoshida2023reversing,zhu2024reversing} and complex conjugation \cite{miyazaki2019complex}. With no analogy to states these tasks may differ for example by the order in which the transformation to be used and those to be created are available. A single use of a transformation can be modified by a chosen (pre-) or (post-)transformation 
	or also by using an additional ancillary system interconnecting them. 
	
	These new possibilities make a difference for example in the discrimination tasks. 
	Two unitary transformations can be non-orthogonal, i.e. not perfectly distinguishable in a single use, but in contrast to states, they can be made perfectly distinguishable if they are used sufficiently many times. On the other hand, similarly to states, completely unknown unitary transformation cannot be cloned. In fact there are not many tasks for just two possible unitary transformations, which would be already studied despite their obviously fundamental role. The goal of this paper is to partially fill this gap.

	\begin{figure*}[t]
		\begin{center}
			\includegraphics[width=1\linewidth
			]{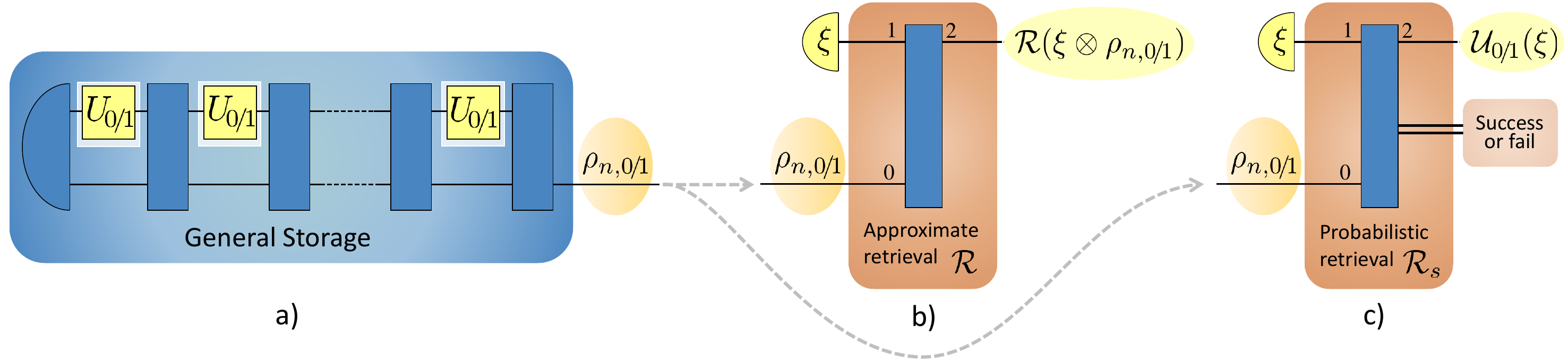}
			\caption{
   Illustration of a general storage and retrieval protocol for two $d$-dimensional unitary transformations $U_0$, $U_1$. These operations are used $n$ times during the storage phase (a). Retrieval can be deterministic (b) and hence yielding only approximation to $\mathcal{U}_{0\!/\!1}(\xi)$, or probabilistic (c), yielding perfect application of the stored unitary to any state $\xi$. 
            }
			\label{fig:gen_scheme01}
		\end{center}
	\end{figure*}

	In this paper, we focus on the problem of storage
	and retrieval \cite{PhysRevA.81.032324,PhysRevLett.122.170502,PhysRevA.106.052423,grosshans2024multicopy}  of an unknown unitary
	transformation that is randomly selected from a
	set of two unitary transformations, say $U_0$ and $U_1$, which cannot be
	perfectly discriminated. The scenario is as
	follows: currently, we are free to use a black box $n$-times, which either performs unitary $U_0$
	or $U_1$. 
	However, later, the black box
	will no longer be available, and we will be asked
	to reproduce its action on some unknown input
	state, say $\ket{\xi}$. 
	If $n$ uses allow to perfectly discriminate between $U_0$ and $U_1$, then the problem trivializes. Therefore, we will from now on assume the non-trivial case in which $U_0$ and $U_1$ cannot be perfectly discriminated with $n$ uses.
	The most general strategy
	we can use is illustrated in Figure
	\ref{fig:gen_scheme01}(a)
	and it consists of two phases:
	\emph{storage} and \emph{retrieval}. In the
	storage phase, we run a quantum circuit that makes
	$n$ calls to the black box
	and we store 
	the resulting
	output state $\rho_{n,i}$ in a quantum
	memory. In the retrieval phase, once the new input
	$\ket{\xi}$  becomes available, we feed both
	$\rho_{n,i}$  
	and $\ket{\xi}\bra{\xi}$ 
	into a retrieving channel $\mathcal{R}$, which should emulate the action of the unknown
	unitary on 
	the state
	$\ket{\xi}$.	
	Storage and retrieval can also be interpreted as a
	machine learning task \cite{wittek2014quantum,schuld2015introduction,Biamonte:2017aa,cerezo2022challenges}, where the
	storage phase corresponds to the training of a 
	quantum network.
	
	We address the optimization of storage and
	retrieval according to two different criteria. 
        In
	the \emph{approximate deterministic} case, we require
	that the retrieval works every time (i.e., the
	retrieval is a quantum channel), and the goal is
	to maximize the average process fidelity. In the
	\emph{perfect probabilistic} case, we allow the
	possibility of retrieval failure (i.e., the
	retrieval is a two-outcome quantum instrument),
	but in case of success, the unknown unitary must
	be retrieved perfectly. Here, the goal is to
	maximize the probability of success.
	
	First, we show that regardless of the scenario
	considered in the retrieval phase, the optimal
	storage phase involves preparing the state that is
	optimal for the discrimination of the two unitaries
	\cite{duan2007entanglement}. This state can be
	generated 
	by the sequential application of
	$n$ uses of the unknown unitary without the need for 
	an ancillary system or entanglement.

	Next, we analytically solve the storage and
	retrieval problem in the approximate deterministic
	case for qubit unitaries. We demonstrate that the
	optimal protocol is a ``measure-and-prepare''
	strategy, meaning that quantum memory is not required
	between the storage and retrieval phases. This
	feature was also observed in the case where the
	unknown unitary is randomly drawn from a
	group \cite{PhysRevA.81.032324}.
	
	We then derive the optimal probability of success
	for the perfect probabilistic case.  We observe
	that in the regime in which the fidelity between
	the unitaries is 
	big 
	(i.e., far from the
	regime in which $U_0$ and $U_1$ can be perfectly
	discriminated) the probability of success scales as
	$ P_{succ} = 1 - \mathcal{O}(n^{-2}) $.  This result can be compared to that in
	Refs. \cite{PhysRevLett.122.170502,PhysRevA.102.032618},
	where the unknown unitary was randomly drawn from
	a group, and the probability of success scaled as
	$ P_{succ} = 1 - \mathcal{O}(n^{-1}) $.  The
	more accurate prior knowledge in the
	present case (two unitaries vs the entire group $U(d)$
	with uniform distribution) is responsible for the observed
	quadratic improvement. 
	
	In order to make our results more practically applicable, 
	we derive a 
	short quantum circuit which
	realizes the optimal perfect probabilistic storage and
	retrieval of two unknown qubit unitaries. It contains just single-qubit gates, one CNOT gate, and one $3$-outcome qubit POVM.  
	
	
	Finally, 
	we also build a quantum linear optical experiment 
	implementing the proposed scheme. The setup 
	uses CNOT gate consisting of partially polarizing beam splitters (PPBS) and utilizing
	the two-photon interference. 
	The unambiguous measurement part is realized using a Mach-Zehnder-type interferometer 
	made with polarizing beam splitters and wave plates in its arms. 
	We comprehensively characterize the quality of experimentally retrieved quantum operations using quantum process tomography.


	The manuscript is organized as follows. In section
	\ref{sec:stor-retr-two} we provide the analytical
	optimization of the storage and retrieval
	problem. First we derive the optimal storage
	protocol and then, in subsection
	\ref{sec:optimal-retrieval}, we solve the
	optimization of retrieval in the approximate
	deterministic case and in the perfect
	probabilistic case. 
	We connect optimal retrieval with the concept of a programmable quantum processor and we interpret the obtained results in this language in subsection \ref{sec:processor}. 
	Subsection \ref{sec:qcircuit} is devoted to 
	deriving a quantum circuit for the realization of
	optimal storage and retrieval protocols in the
	qubit case.  In section \ref{sec:phot-demonstr} we present 
	a proof of
	principle quantum optics experiment that implements the perfect probabilistic version of the optimal storage and retrieval.  Section \ref{sec:discussion}
	contains a discussion of the obtained theoretical
	and experimental results. Some technical parts of
	proofs and experiment description are placed in
	the appendixes.

	\section{Storage and retrieval of two unitary
		channels}
	\label{sec:stor-retr-two}
	The unknown
	$d-$dimensional unitary transformation that we are
	supposed to use only $n$ times during the storage
	phase is chosen with equal prior probability from
	two options, which are denoted
	as $U_0$ and $U_1 $ in
	the rest of the manuscript.
	We will 
	denote by $\mathcal{U}_i(.)= U_i \, . \, U_i^\dagger$ the
	action of the given unitary on the density
	operators.
	
	If $U_0$ and $U_1$
	are perfectly distinguishable with $n$ uses, then
	the storage and retrieval
	problem trivializes. In particular, in the
	storage phase we can exactly discriminate
	between $U_0$ and $U_1$. Then,
	depending on the outcome, in the retrieval
	phase we prepare the unknown unitary. 
	Therefore, we will assume in the rest of the manuscript that the unitaries $U_0$ and $U_1$
	cannot be perfectly distinguishable by their $n$ uses.
	
	If only one use of the unknown
	transformation is available (i.e. $n=1$) the
	storage phase consists of applying the single use
	of the unitary to a state $\rho$. On the other
	hand, if many uses are available (i.e. $n > 1$)
	the most general storage phase consists of
	inserting the $n$ uses of the unknown unitary
	transformation in the open slots of a quantum
	circuit (see Figure~\ref{fig:gen_scheme01}a). At the
	end of the storage phase, we obtain one of two
	states $\rho_{n,0}$, $\rho_{n,1}$
	depending on the identity of the stored transformation.

	A quantum circuit implementing the storage phase is
	generally a network of quantum transformations
	that includes channels and measurements. We can always dilate each component of the network
	to be an isometric transformation \cite{PhysRevA.80.022339}. Therefore, we can assume without loss of generality that the
	states obtained at the end of the storage
	phase are pure
	states 
	$\rho_{n,0}=\ket{\psi_{n,0}}\bra{\psi_{n,0}}$,
	$\rho_{n,1}=\ket{\psi_{n,1}}\bra{\psi_{n,1}}$.
	For any pair of states $\ket{a_0}$ and
	$\ket{a_1}$ with $|\braket{a_0}{a_1}| \geq
	|\braket{\psi_{n,0}}{\psi_{n,1}}|$ one can show
	that there exists a channel
	$\mathcal{C}$ such that
	$\mathcal{C}(\ketbra{\psi_{n,i}}{\psi_{n,i}}) =
	\ketbra{a_i}{a_i}$.  
	As a consequence the optimal storage phase is the
	one for which the scalar product
	$|\braket{\psi_{n,0}}{\psi_{n,1}}|$ is minimum.
	However, from the results on the optimal
	discrimination of unitaries
	\cite{acin2001statistical,duan2007entanglement,chiribella2008memory}
	we know that $|\braket{\psi_{n,0}}{\psi_{n,1}}|
	\geq \cos(2n\alpha)> 0 $ 
	and $4 \alpha
	$ corresponds to the length of the smallest arc
	containing all the eigenvalues of $U_1^{\dag } U_0
	$ on the unit circle.  Note that
	$\cos(2n\alpha) > 0$ (or equivalently $4n\alpha<
	\pi$), because we assumed that $U_0$ and
	$U_1$ cannot be perfectly discriminated with
	$n$ uses.
	
	By suitable fixed pre- or
	post-processing unitary transformations $V$ and $W$, we
	can transform the storage and retrieval protocol for
	$\{U_0,U_1\}$ into a protocol for
	$\{W U_0V, W U_1V \}$ with
	the same performance. The same holds for the discrimination of unitaries. Moreover, the global phase of the unitary operator is irrelevant. Thus, without loss of
	generality, we can assume that \footnote{Consider choosing $V=U_0^\dagger$, $W=(U_1 U_0^\dagger)^{-1/2}$.} 
	\begin{align}
	\label{eq:26}
	\begin{aligned}
	U_0 &= e^{i\alpha} \ketbra{0}{0} +
	\sum_{k = 1}^{d-2}e^{i\beta_k} \ketbra{k}{k} +
	e^{-i\alpha} \ketbra{d-1}{d-1}\\
	U_1 &= e^{-i\alpha} \ketbra{0}{0} +
	\sum_{k= 1}^{d-2}e^{-i\beta_k} \ketbra{k}{k} +
	e^{i\alpha} \ketbra{d-1}{d-1},
	\end{aligned}
	\end{align}
	where the eigenvalues $e^{i\beta_k}$ 
	are contained in the smallest arc connecting $e^{i\alpha}$ and $e^{-i\alpha}$ and $d$ is dimension. 
	The optimal storage protocol is thus provided
	by the sequential
	strategy derived in
	Ref.\cite{duan2007entanglement} as follows:
	\begin{proposition}[Optimal storage]
		\label{prop:optimalstorage}
		Without loss of generality, let us assume that $U_0$ and $U_1$ are as in
		Equation~\eqref{eq:26}. The optimal storage
		strategy is then given (see Figure (\ref{fig:opt_storage})) by applying $U_i^n$ to the
		input state $\ket{+} :=
		\frac{1}{\sqrt{2} }\big (   \ket{0} + \ket{d-1}
		\big)$, i.e.
        $\ket{\psi_{n,i}} = U_i^n\ket{+}.$ 
	\end{proposition}

	\begin{figure}[t]
		\begin{center}
			\includegraphics[width=7cm]
                {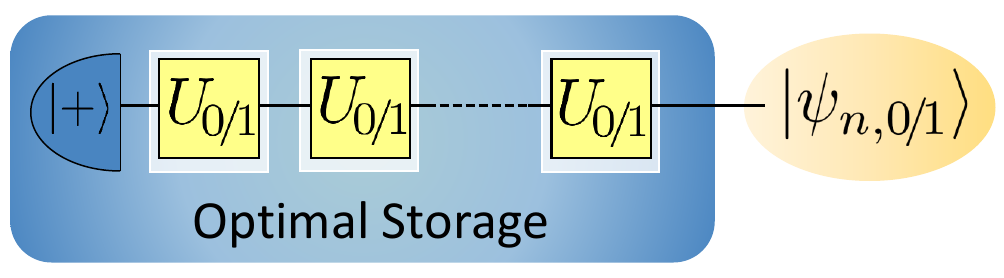}
			\caption{
                Optimal quantum circuit for storage of two d-dimensional unitary transformations $U_0$,$U_1$ used $n$ times. The circuit is the same for both considered types of retrieval.
                }
			\label{fig:opt_storage}
		\end{center}
	\end{figure}

	Let us now consider the retrieval phase. The
	input of this phase is the tensor product between
	the state $ \ket{\psi_{n,i}}$, which encodes
	information about the unknown unitary $U_{i}$, and
	the state $\xi$, to which the retrieved
	transformation should be applied. The result of
	this action should ideally be
	$\mathcal{U}_i(\xi)$.  We can now approach the
	problem in at least two different ways.
	
	In the \emph{approximate deterministic} case, the
	retrieval is a quantum channel $\mathcal{R}$,
	which maps bipartite input into single-partite
	output. Within this approach, the goal is to find
	the optimal $\mathcal{R}$ such that
	$\mathcal{R}(\ketbra{\psi_{n,i}}{\psi_{n,i}}\otimes\xi)$
	is as close as possible to $\mathcal{U}_i(\xi)$.
	For the sake of clarity, it is convenient to introduce
	the following labeling of the Hilbert spaces (see
	Figure~\ref{fig:gen_scheme01}b). We denote with
	$\mathcal{H}_0$ the system carrying the output of
	the storage phase
	($\ket{\psi_{n,i}} \in \mathcal{H}_0$), with
	$\mathcal{H}_1$ the system carrying the input
	$\xi$ and with $ \mathcal{H}_2$ the output system
	of the retrieval phase, i.e.
	$\mathcal{R} : \mathcal{L}(\mathcal{H}_0 \otimes
	\mathcal{H}_1) \to \mathcal{L}(\mathcal{H}_2)$.
	In order to assess the quality
	of the retrieval we choose as a criterion the
	average of state fidelity over pure input states,
	namely
	\begin{align}
	\label{eq:3}
	F_{avg} &:=  \int \! \! \mathrm{d}\phi  \sum_{i=0,1} \frac12
	f\Big(\mathcal{R}(\psi_{n,i}\otimes
	\ketbra{\phi}{\phi}),
	\mathcal{U}_i(\ketbra{\phi}{\phi})\Big) \\
	\psi_{n,i} &:=\ketbra{\psi_{n,i}}{\psi_{n,i}} \nonumber
	\end{align}
	where $\mathrm{d}\phi$ is the Haar measure, and
	$f(\cdot, \cdot)$ is the state fidelity. 
	Equation~\eqref{eq:3} can be
	rewritten 
	\cite{raginsky2001fidelity,NIELSEN2002249}
	in terms of quantum process fidelity
	as follows:
	%
	%
	\begin{align}
	F_{avg} &= \frac{1}{d+1}+\frac{d}{d+1}F_{e}\\
	F_e &\equiv\Tr [R D], \label{eq:4}\\
	D &:= \frac{1}{2 d^2} 
	\sum_{i=0,1}\ketbra{\psi^*_{n,i}}{\psi^*_{n,i}}
	\otimes \KetBra{U_i}{U_i} \label{eq:5}  \\
	\ket{\psi^{*}_{n,i}} &:= U_i^{*n}\ket{+},
	\end{align}
	where
	$R=(\mathcal{I} \otimes \mathcal{R} )
	\KetBra{I}{I} \in
	\mathcal{L}(\mathcal{H}_0\otimes
	\mathcal{H}_1\otimes \mathcal{H}_2)$
	($\KetBra{I}{I} \in \mathcal{L}(\mathcal{H}_0\otimes
	\mathcal{H}_1\otimes \mathcal{H}_0\otimes \mathcal{H}_1$)) is the
	Choi-Jamiolkovski operator
	\cite{choi1975completely,jamiolkowski1972linear}
	of the channel $\mathcal{R}$,
	$\KetBra{U_i}{U_i} \in
	\mathcal{L}(\mathcal{H}_1\otimes \mathcal{H}_2)$
	is the Choi-Jamiolkovski operator of the
	unitary channel $\mathcal{U}_i$, we used the
	notation \cite{PhysRevA.80.022339}
	$\Ket{A} := \sum_{i,j} A_{i,j} \ket{j}\ket{i} = (I \otimes A ) \Ket{I}$,
	for an operator
	$A = \sum_{i,j} A_{i,j} \ket{i}\bra{j} $, and
	$A^*$ denotes the complex conjugate. The
	isomorphism $A \leftrightarrow \Ket{A}$, the
	Choi-Jamiolkovski isomorphism, the transposition
	$A^T$, and the complex conjugation $A^*$ are
	defined with respect to some fixed orthonormal
	basis $\{\ket{i}\}_{i=0}^{d-1}$ which is the same
	basis in which $U_0$ and $U_1$ take the form of
	Equation~\eqref{eq:26}.
	The normalization
	factor guarantees that $0 \leq F_e \leq 1$.  We are
	then left with the following optimization problem for the \emph{approximate deterministic retrieval}:
	\begin{align}
	&\begin{aligned}\label{eq:optprob1}
	F_{e}=\; & \underset{R}{\text{maximize}}
	& & \Tr[RD]\\
	& \text{subject to}
	& & R\geq 0 , \quad \Tr_2[R] = I,
	\end{aligned}
	\end{align}
	where $ \Tr_2$ is the partial trace on the
	Hilbert space $\mathcal{H}_2$.
	Notice that we could loosen the constraint by requiring
	that $R$ should be a quantum operation ($\Tr_2[R]
	\leq I$) and not
	necessarily a quantum channel (i.e. $\Tr_2[R] = I$). 
	It is easy to verify that this change is
	irrelevant:  even if we optimize over
	the larger set of quantum operations, the optimal
	solution will still be a quantum channel. 
	Indeed, let $R$ be the optimal quantum operation.
	Then there exists a quantum channel  $S$ such that
	$R \leq S$ and $\Tr[DR] \leq \Tr[DS] $ because $D$
	is a positive operator.
	
	Another approach is
	\emph{perfect probabilistic} retrieval. 
	In this case, the retrieval is described by a two-outcome quantum
	instrument $\mathcal{R}=\{\mathcal{R}_s,
	\mathcal{R}_f\}$, where the index denotes success ($s$) or failure ($f$). If the classical outcome is
	$s$, we demand the exact retrieval of the unknown
	unitary, i.e.
	$\mathcal{R}_s(\psi_{n,i} \otimes \xi ) = \lambda_{\xi,i}\,\mathcal{U}_i (\xi)$  for any $\xi$ and we define $\lambda_{\xi,i} := \Tr[\mathcal{R}_s(\psi_{n,i}
	\otimes \xi )] $
	(we remind
	that $\psi_{n,i} :=
	\ketbra{\psi_{n,i}}{\psi_{n,i}}$).
	The factor $ \lambda_{\xi,i}$ is the probability
	of success of perfect retrieval when the unknown
	unitary is $U_i$ and the input state is $\xi$.
	It is easy to observe that such a
	probability of success does not depend on $\xi$,
	namely $\lambda_{\xi,i}=
	\lambda_i$ for any $\xi$.
	Indeed, if we had
	$\lambda_{\xi,i} \neq \lambda_{\xi',i} $ for a
	pair $\xi$, $\xi'$ of states, this would imply
	that
	$  \mathcal{R}_s(\psi_{n,i}
	\otimes (\xi + \xi') ) =
	\lambda_{\xi,i}\mathcal{U}_i(\xi) +
	\lambda_{\xi',i}\mathcal{U}_i(\xi') = \mathcal{U}_i(\lambda_{\xi,i}\xi +
	\lambda_{\xi',i}\xi') \not\propto
	\mathcal{U}_i(\xi + \xi') $.
	The perfect retrieval condition can therefore be
	written as follows:
	\begin{align}
	\label{def:psrcond1}
	\bra{\psi^*_{n,i}} R_s \ket{\psi^*_{n,i}} = \lambda_i
	\KetBra{U_i}{U_i}   \quad i=0,1,
	\end{align}
	where $ R_s$ is the Choi-Jamiolkovski operator of
	the quantum operation $R_s$.
	As a criterion to assess the performance of the
	retrieval strategy we choose 
	the average
	probability of success, which in equation reads
	\begin{align}
	\label{eq:9psuc}
	P_{succ} := \frac{1}{2} (\lambda_0+\lambda_1).
	\end{align}
	It is easy to verify that
	$\frac{1}{2} (\lambda_0+\lambda_1) = \Tr[R_s D]$,
	where $D$ was defined in Equation~\eqref{eq:5}.
	Reminding that $R_s$ is a quantum operation if and
	only if $R_s \geq 0$ and $\Tr_2[R_s] \leq I$,
	we obtain the following optimization
	problem for the \emph{perfect  probabilistic retrieval}:
	\begin{align}
	&\begin{aligned}\label{eq:optprob2}
	P_{succ}= \; & \underset{R_s}{\text{maximize}}
	& & \Tr[R_sD]\\
	& \text{subject to}
	& & R_s\geq 0 , \quad \Tr_2[R_s] \leq I,\\
	&& &R_s \mbox{ obeys Equation~\eqref{def:psrcond1}}.
	\end{aligned}
	\end{align}

	It is worth noticing that the optimization
	problems~\eqref{eq:optprob1}, \eqref{eq:optprob2} are almost the same, the only difference being the perfect retrieval condition~\eqref{def:psrcond1}.

	\subsection{The optimal retrieval}
	\label{sec:optimal-retrieval}

	In this section we will find an analytical 
	solution to the optimization problems of Equation~\eqref{eq:optprob1} and \eqref{eq:optprob2} when the unknown unitaries
	act on a two-dimensional Hilbert space
	(i.e. $U_0$ and $U_1$ are single-qubit
	gates). Later, we will present the analytical
	solution to the optimization problem for 
	perfect probabilistic retrieval for general
	dimension $d$. 
	If $d=2$ operators $U_i$ and states $\ket{\psi_{n,i}}$ 
        simplify as follows:
	\begin{align}
	& \begin{aligned}
	U_0 = e^{i \alpha } \ketbra{0}{0} +
	e^{-i \alpha } \ketbra{1}{1} \\
	U_1 = e^{-i \alpha } \ketbra{0}{0} +
	e^{i \alpha } \ketbra{1}{1}
	\end{aligned}  \\
	&\begin{aligned}
	\ket{\psi_{n,0}} &=    \frac{1}{\sqrt{2} }\big (
	e^{i \alpha n} \ket{0} + e^{-i \alpha n}\ket{1}
	\big) \\
	\ket{\psi_{n,1}} &=   \ket{\psi^*_{n,0}} =
	\sigma_x \ket{\psi_{n,0}},
	\quad \sigma_x =
	\ketbra{0}{1} + \ketbra{1}{0}.
	\end{aligned}
	\end{align}
	We 
	remind that 
	$4n \alpha < \pi$ since
	otherwise the unitaries become perfectly distinguishable.
	
	We now observe that the following commutation
	relations hold:
	\begin{align}
	&  \begin{aligned}
	[D , W( \beta,\gamma,l) ]= 0, \quad \forall
	\beta,\gamma \in \mathbb{}\mathbb{R}, \;\forall l \in
	\{0,1\} , \\
	\end{aligned}
	\\
	&
	\begin{aligned}
	& W(\beta,\gamma,l) :=
	\sigma^{(l)} \otimes
	\sigma^{(l)}Z(\beta,\gamma) \otimes
	\sigma^{(l)}Z^*(\beta,\gamma)  \\
	&Z(\beta,\gamma)     := e^{i\beta} \ketbra{0}{0} +
	e^{i \gamma} \ketbra{1}{1},\\
	& \sigma^{(0)} := I, \quad \sigma^{(1)} := \sigma_x.
	\end{aligned}\nonumber
	\end{align}
	We notice that $W(\beta,\gamma,l)$ is a unitary
	representation of the compact Lie group
	$ G := (\mathsf{U}(1) \times
	\mathsf{U}(1)) \rtimes \mathbb{Z} $,  the
	semidirect product is defined by the left action
	$\eta: l \mapsto \eta_{l}$,
	$\eta_{0} (e^{i\beta},
	e^{i\gamma}) = (e^{i\beta}, e^{i\gamma})$,
	$\eta_{1} (e^{i\beta},
	e^{i\gamma}) = ( e^{i\gamma}, e^{i\beta})$. 
	In order to lighten the notation, we will then
	write $W(g)$, $g \in G$ in place of $W(\beta,\gamma,l)$.
	Exploiting this symmetry we can prove the
	following lemma:
	
	\begin{lemma}[Symmetric retrieval is
		optimal]\label{lmm:decomposition-r} Without loss
		of generality, we can assume that the optimal
		approximate deterministic retrieval $R$
		(respectively the optimal perfect probabilistic retrieval
		$R_s$) satisfies $[R, W(g)] =0 $ (respectively
		$[R_s , W(g)] =0 $) for any $g \in G $.
		Moreover, if $[R_s , W(g)] =0 $ then Equation~\eqref{def:psrcond1} holds with
		\begin{align}
		\label{eq:8}
		\lambda_0 = \lambda_1 = \lambda.
		\end{align}
	\end{lemma}
	\begin{proof}
		The proof follows the Holevo’s averaging
		argument \cite{holevo2011probabilistic}.
		
		Let
		$R$ be the optimal approximate deterministic
		retrieval.  If $R$ is a quantum operation we
		observe that also $W^\dag(g) R W(g)$ is a
		quantum operation and therefore also the average
		$\overline{R} = \int_G dg W^\dag(g) R W(g)$ is a
		quantum operation. Clearly $[\overline{R} ,
		W(g)] = 0 $ for any $g \in G$.
		Since $[D , W(g)] = 0 $ for any $g \in G$ we
		have that $\Tr [RD] = \Tr [\overline{R} D]$,
		i.e. $\overline{R}$ is another optimal approximate
		deterministic retrieval.
		
		The thesis for the perfect  probabilistic case
		follows from the fact that if
		an optimal $R_s $ obeys the perfect retrieval condition  of  Equation~\eqref{def:psrcond1} with the values $\lambda_0$
		and $\lambda_1$ then
		the averaged operator $\overline{R}_s$
		obeys the perfect retrieval condition  of  Equation~\eqref{def:psrcond1} with $\overline{\lambda}_0 =
		\overline{\lambda}_1 = \frac12 (\lambda_0 +
		\lambda_1)$. 
	\end{proof}

	Thanks to this result, we can conveniently express
	the optimal retrieval in a block diagonal form.
	First, let us now define the following projectors:
	\begin{align}
	\label{eq:10} &
	\begin{aligned} &P := \frac12 \Big
	(\KetBra{e_1}{e_1} + \KetBra{e_2}{e_2} \Big ), \\
	&P':= \frac12 \Big ( \KetBra{e'_1}{e'_1} +
	\KetBra{e'_2}{e'_2} \Big ) ,
	\end{aligned}\\ &
	\label{eq:defee}
	\begin{aligned} \Ket{e_1} &:=
	\ket{+}\Ket{I},\qquad &\Ket{e_2}
	&:=\ket{-}\Ket{\sigma_z},\\ \Ket{e'_1} &:=
	\ket{+}\Ket{\sigma_z},\qquad &\Ket{e'_2}
	&:=\ket{-}\Ket{I}.
	\end{aligned}
	\end{align}
	Notice that $P'= ( I \otimes
	\sigma_z) P ( I\otimes \sigma_z )$ and
	$\Ket{e'_i} :=(I \otimes \sigma_z)\Ket{e_i}$.
	We can now prove the following lemma.
	\begin{lemma}[Block diagonal
		retrieval]\label{lmm:blockdiagretriev}
		Without loss
		of generality, we can assume that the optimal
		approximate deterministic retrieval $R$ and the
		optimal
		perfect probabilistic retrieval  $R_s$
		satisfy $R = PRP + P'RP'$
		and $ R_s = PR_sP + P'R_sP'$ respectively.
		Furthermore,  we have
		\begin{align}
		\label{eq:1}
		\begin{aligned}
		\bra{\psi^*_{n,i}}PR_sP\ket{\psi^*_{n,i}} &=
		\lambda_A \KetBra{U_i}{U_i}  \\
		\bra{\psi^*_{n,i}}P'R_sP'\ket{\psi^*_{n,i}}& =
		\lambda_B \KetBra{U_i}{U_i}.
		\end{aligned}
		\end{align}
		for any $i=0,1$.
	\end{lemma}
	\begin{proof} One can verify that $P$ and $P'$ are
		projectors on invariant subspaces of the
		representation $W(g)$. Since from
		Lemma~\ref{lmm:decomposition-r} we have that
		$[R,W(g)] = 0$ holds, we have
		that $R =  PRP + P'RP' + QRQ$ where
		$Q := I - P - P'$ (the same holds for $R_s$).
		Since $D = PDP + P'DP' $ the thesis for
		the approximate deterministic case follows.
		In the perfect probabilistic case, 
		from inserting $R_s =  PR_sP + P'R_sP' + QR_sQ$ into Equation (\ref{def:psrcond1})  we conclude that
		\begin{align*}
		\bra{\psi^*_{n,i}} PR_sP \ket{\psi^*_{n,i}}
		&= \lambda_A \KetBra{U_i}{U_i} \\
		\bra{\psi^*_{n,i}} P'R_sP' \ket{\psi^*_{n,i}}
		&=\lambda_B \KetBra{U_i}{U_i}\\
		\bra{\psi^*_{n,i}} QR_sQ\ket{\psi^*_{n,i}}
		&=\lambda_C \KetBra{U_i}{U_i},
		\end{align*}  
		because 
		terms on the left 
		are positive operators, which are supposed to sum up to a rank one operator $\lambda \KetBra{U_i}{U_i}$.
		Next, we observe that $Q=I\otimes(\KetBra{\sigma_x}{\sigma_x}
		+\KetBra{\sigma_y}{\sigma_y})\equiv I \otimes Q'$, 
		$Q'\Ket{U_i}=0$ and using Eq.~(\ref{def:psrcond1}) we obtain
		\begin{align*}
		\bra{\psi^*_{n,i}}
		QR_sQ\ket{\psi^*_{n,i}} &=Q'\bra{\psi^*_{n,i}}
		R_s\ket{\psi^*_{n,i}}Q' \nonumber \\
		&=\lambda_i  Q'\KetBra{U_i}{U_i}Q' = 0.
		\end{align*}
		Thus, $\bra{\psi^*_{n,i}}\otimes \bra{\varphi} Q R_{s} Q \ket{\psi^*_{n,i}}\otimes \ket{\varphi} =0$ for $i=0,1$ and $\forall \ket{\varphi}\in \hilb{H}_{12}$. In other words, support of positive semidefinite operator $Q R_{s} Q$ has to be orthogonal to the above vectors, which span the whole space. We conclude that the support of $Q R_{s} Q$ is empty, i.e. $Q R_{s} Q=0$ and the thesis follows.
	\end{proof}
	
	From the block diagonal structure of Lemma
	\ref{lmm:blockdiagretriev}, we have a convenient
	way to rephrase the optimization of the retrieval.
	\begin{proposition}
		\label{prp:optimal-approxdet_retrieval_compactform}
		The optimization
		of the approximate deterministic retrieval (Equation~\eqref{eq:optprob1}) is equivalent to the following state discrimination problem:
		\begin{align}
		&      \begin{aligned}
		&\begin{aligned}
		F_{e}=\;     & \underset{A,B}{\mathrm{maximize}}
		& & \bra{u}A\ket{u} + \bra{v}B \ket{v}\\
		& \mathrm{subject}\; \mathrm{to}
		& & 0 \leq A,B \leq I, \quad A+B \leq I.
		\end{aligned}
		\end{aligned}
		\label{eq:optprob_approxdet_final}
		\end{align}
		The optimization
		of the perfect probabilistic retrieval~\eqref{eq:optprob2} is equivalent to the following unambiguous state discrimination problem:
		\begin{align}
		& \begin{aligned}
		&\begin{aligned}
		P_{succ}=\; & \underset{A,B}{\mathrm{maximize}}
		& & \bra{u}A\ket{u} + \bra{v}B \ket{v}\\
		& \mathrm{subject}\;\mathrm{to}
		& & 0 \leq A,B \leq I \quad  A+B \leq I \\
		&&&\bra{v}A \ket{v} = \bra{u}B \ket{u} = 0.
		\end{aligned}	\end{aligned}\label{eq:optprob_perfprob_final}
		\end{align}   
		where $A,B $ are a $2\times 2$ matrices
		and we defined
		\begin{align}
		\label{eq:2}
		&\ket{u} :=
		\begin{pmatrix} c_n c\\ s_n s
		\end{pmatrix} \quad \ket{v} :=
		\begin{pmatrix} c_n s\\ -s_n c
		\end{pmatrix} \\
		& \begin{aligned}
		s&:= \sin(\alpha) ,  &c:= &\cos(\alpha)\\
		s_n&:= \sin(n\alpha) , & c_n :=&\cos(n\alpha) .
		\end{aligned} \nonumber
		\end{align}
	\end{proposition}
	\begin{proof}
		The proof is reported in Appendix \ref{sec:appqrsq}.
	\end{proof}
	
	The optimal discrimination of a pair of
	quantum states has a known analytical solution
	both in the minimum error and in the
	unambiguous case.  The optimal value of the
	process fidelity
	for the approximate
	deterministic case reads:
	\begin{align}
	\label{eq:7}
	F_e = \frac{1}{2} + \frac{1}{2}\sqrt{1-(\sin{2\alpha}\cos{2n\alpha})^2}.
	\end{align}
	The optimal strategy can be realized by the
	measure-and-prepare protocol presented in Figure \ref{fig:meas_prep} 
	(see    Appendix \ref{sec:real-optim-appr} for the details).  
	Since classical information can be cloned, our result  can be extended to the case where we
	must reproduce $m \geq 1$ copies of the unknown unitary with the criterion of maximizing the single-copy fidelity.

	
	\begin{figure}[t]
		\begin{center}
			\includegraphics[width=\columnwidth
			]{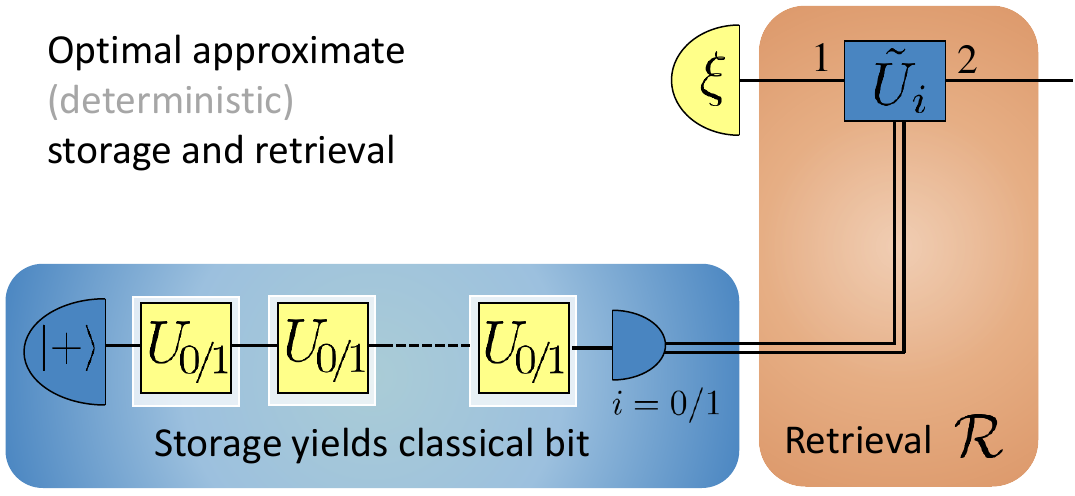}
			\caption{Optimal approximate deterministic storage and retrieval of two $d$-dimensional unitary transformations can be realized via optimal discrimination among the two unitaries in the storage phase and conditional preparation of slightly modified unitaries in the retrieval phase.
			}
			\label{fig:meas_prep}
		\end{center}
	\end{figure}

	The optimal solution to unambiguous state
	discrimination is known
	\cite{jaeger1995optimal,SedlakActaUnambigous}
	and it gives
 	\begin{align}
	\label{eq:fom_final} &P_{succ} =
	\left \{
	\begin{aligned}
	&\frac{1-(\cos{2n\alpha})^2}{2(1-\cos{2n\alpha}\cos{2\alpha})}&& \mbox{ if }
	\alpha \in (0,{\chi_n}) \\ 
	&1-\cos{2n\alpha}\sin{2\alpha}
	&& \mbox{ if } \alpha \in [{\chi_n}, \frac{\pi}{4n}] 
	\end{aligned} 
 \right.
 \end{align}
where $\chi_n$ is the solution of the following trigonometric equation
\begin{align}
\cos{(2n\alpha)} = \frac{1}{\cos{(2\alpha)}+ \sin{(2\alpha)}}.
\label{eq:chi_alfa}
\end{align}
The range of angles $\alpha$ for which $\alpha$ is below (above) angle $\chi_n$ we denote as small (large) $\alpha$ regime, respectively. 
	Success probability for various values of
	$n,\alpha$ is illustrated on
	Figure~\ref{fig:fig1}.  
        For $\alpha \ll 1$
        the fidelity between $U_0$ and $U_1$ is
	large. Therefore, $U_0$ and $U_1$ can be perfectly
	discriminated only with a large number of uses,
	i.e. ${n\gg1}$. Therefore it makes sense to power
	expand Equation~\eqref{eq:fom_final} around
	$\alpha =0$ and study how the probability of success scales with $n$. We obtain 
	\begin{align}
	\label{eq:16a}
	P_{succ} = 1 - \frac{1}{n^2+1} +\frac{n^2+2n^4-3n^6}{3(n^2+1)^2}\alpha^2 + \mathcal{O}(\alpha^4).
	\end{align}
	In comparison, the probability of success of the perfect
	probabilistic storage and retrieval of an unknown
	unitary transformation which is randomly picked from unitary group $U(d)$ with Haar measure reads \cite{PhysRevLett.122.170502,PhysRevA.102.032618}
	\begin{align}
	\label{eq:16b}
	&P_{succ} =
	\left \{
	\begin{aligned}
	&
	1 - \frac{d^2-1}{n+d^2-1}
	&&\!\!\mbox{ if }
	U\in U(d),\, d \geq 2
	\\
	&1 - \frac{1}{n+1}
	&&\!\! \mbox{ if } 
	U= 
 \begin{pmatrix}
 1 & 0 \\
          0&e^{i\varphi}
          \end{pmatrix} . 
	\end{aligned}  \right.
	\end{align}
	Therefore, a more accurate prior information about
	the unknown unitary allows us to achieve a quadratic
        improvement. 
	
	Equation~\eqref{eq:fom_final} provides the optimal probability of success for the perfect probabilistic storage and retrieval of
	a pair of unitaries also in arbitrary dimension $d$.
	Indeed, let $U_0$ and $U_1$
	be given as in Equation~\eqref{eq:26}.
	By restricting the action of $U_0$ and $U_1$
	on the subspace spanned by $\ket{0}$
	and $\ket{d-1}$ we have the two single-qubit unitaries
	$\tilde{U}_0   :=
	e^{ i \alpha} \ketbra{0}{0} +
	e^{- i \alpha} \ketbra{d-1}{d-1}$ and
	$  \tilde{U}_1   :=
	e^{ -i \alpha} \ketbra{0}{0} +
	e^{ i \alpha} \ketbra{d-1}{d-1} $.
	On one hand, the probability of success for the perfect storage and retrieval
	of the pair $\{ \tilde{U}_0 , \tilde{U}_1 \}$
	is given by Equation~\eqref{eq:fom_final}
	and it must be an upper bound to the
	the probability of success for the perfect storage and retrieval
	of the pair $\{ U_0 , U_1 \}$.
	On the other hand, we can provide a retrieval strategy that achieves the bound  (see Appendix \ref{sec:isom-appr-perf} for the details).

	\begin{figure}[t]
		\begin{center}
			\includegraphics[width=\columnwidth
			]{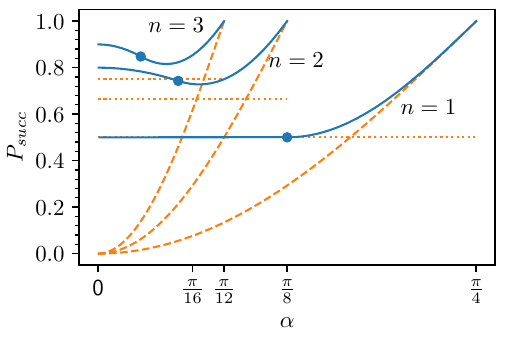}
			\caption{Optimal success probability of perfect storage
				and retrieval of two $d$-dimensional unitary
				transformations as a function of $n$ - the number
				of black box uses and of the angle $\alpha$,
				which equals half of the angular spread of
				$(U_1)^\dagger U_0$. Dotted and dashed lines represent
				performance achievable by non-optimal strategies
				based on results from Ref. \cite{PhysRevA.102.032618} and
				based on unambiguous discrimination, respectively. The blue circles on the
				lines mark the transition between small/large $\alpha$
				regime of the success probability.}
			\label{fig:fig1}
		\end{center}
	\end{figure}


	\subsection{Quantum processor for $2$ unitary transformations}
	\label{sec:processor}
	
	Retrieval phase can be also viewed as a
	programmable quantum processor, since the state
	resulting from the storage phase programs
	(encodes) the transformation to be performed on
	another system. 
	In the rest of this section we aim to rewrite the obtained analytical formulas in the form more convenient from the perspective of quantum processors. 
	The only relevant parameter for the two possible program states (previously denoted as $\ket{\psi_{n,0}}$, $\ket{\psi_{n,1}}$) is their mutual angle, which we parameterize through the overlap as 
	$|\braket{\psi_{n,0}}{\psi_{n,1}}|=\cos{\beta}$ $\beta\in [0,\pi/2]$, i.e. $\beta=2n \alpha$, since the optimization of the retrieval leads to the same formulas also if $n$ is considered to be any non-negative number.
	In this notation, quantum proccess fidelity, Equation~(\ref{eq:7}), can be rewritten as
	\begin{align}
	\label{eq:det_fom_proc}
	F_e = \frac{1}{2} + \frac{1}{2}\sqrt{1-(\sin{2\alpha}\cos{\beta})^2}.
	\end{align}
	The corresponding optimal trade-off between the angle of the program states $\beta$, angular spread ($4\alpha$) of the eigenvalues of the relative unitary $U_1^{\dag } U_0$, and the achievable quantum process fidelity $F_e$ is depicted in Figure~\ref{fig:opt_approx_proc}.
	
	\begin{figure}[t]
		\begin{center}
			\includegraphics[width=8.1cm]{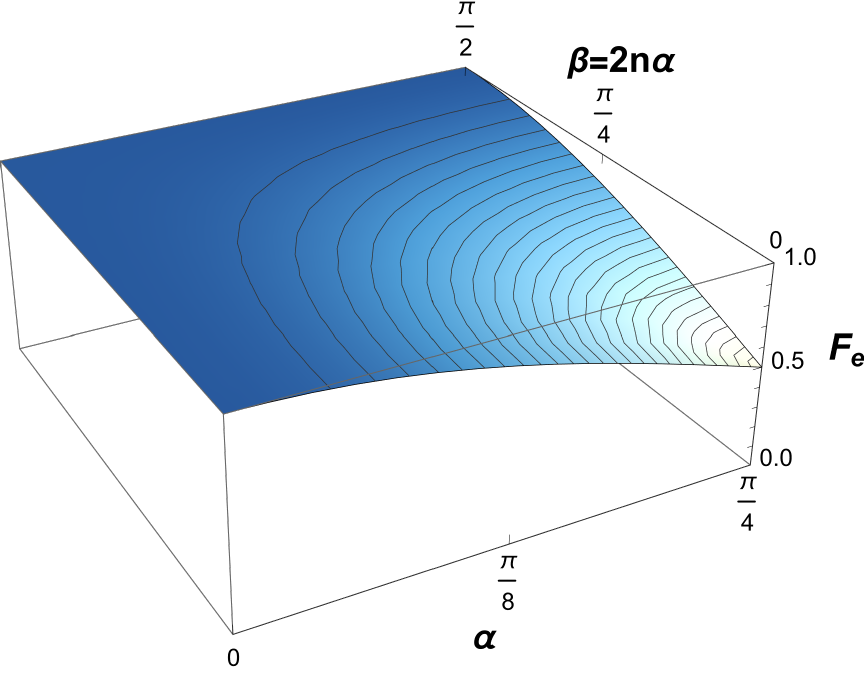}
			\caption{Maximum achievable quantum process fidelity $F_e$ in the \emph{approximate deterministic 
            retrieval
            } 
            as a function of the angle of the program states $\beta$ and angle $\alpha$ representing one fourth of the angular spread of the eigenvalues of the relative unitary $U_1^{\dag } U_0$ for unitary transformations $U_0$, $U_1$. 
			}
			\label{fig:opt_approx_proc}
		\end{center}
	\end{figure}
	
	We can inspect the unambiguous case in the same way. 
	The success probability for unambiguous state discrimination, Equation~(\ref{eq:fom_final}),  can be rewritten in terms of parameters $\alpha, \beta$ as
	\begin{align}
	\label{eq:fom_proc_final} 
	&P_{succ} =
	\left \{
	\begin{aligned}
	&\frac{1-(\cos{\beta})^2}{2(1-\cos{\beta}\cos{2\alpha})}&& \mbox{ if }
	\beta \leq \beta_B(\alpha)
	\\
	&1-\cos{\beta}\sin{2\alpha}
	&& \mbox{ if } \beta \geq \beta_B(\alpha)
	\end{aligned}  \right.\\
	\nonumber \\
	\label{eq:def_trans_beta}
	& \beta_B(\alpha):=\arccos\frac{1}{\cos{2\alpha}+\sin{2\alpha}}.
	\end{align}
	
	We depict the achievable trade-off between parameters $\alpha$, $\beta$ and success probability in Figure 
	\ref{fig:opt_unamb_proc}.

	\begin{figure}[t]
		\begin{center}
			\includegraphics[width=\columnwidth]{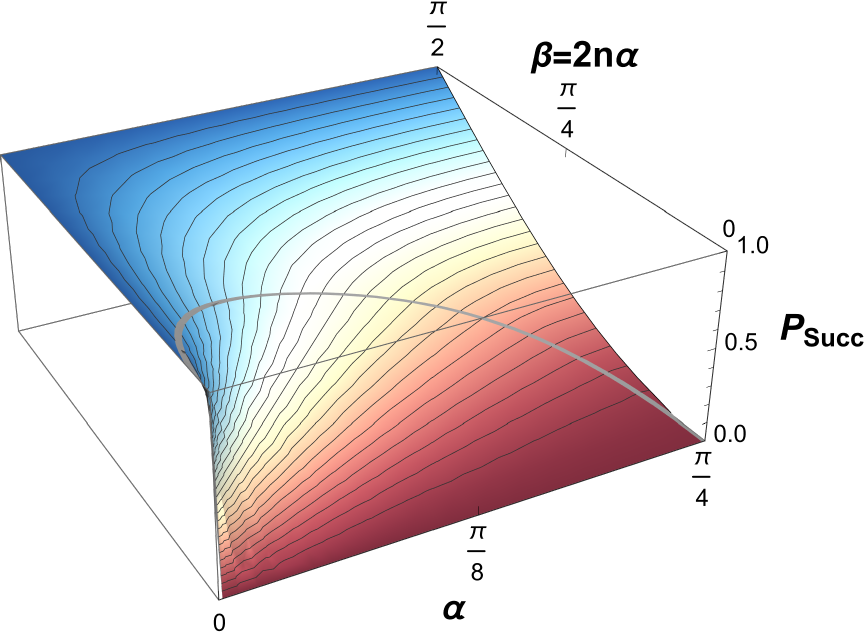}
			\caption{Maximum achievable success probability $P_{succ}$ in the \emph{perfect probabilistic 
            retrieval 
            } 
            as a function of the angle of the program states $\beta$ and angle $\alpha$ representing one fourth of the angular spread of the eigenvalues of the relative unitary $U_1^{\dag } U_0$ for unitary transformations $U_0$, $U_1$.
			}
			\label{fig:opt_unamb_proc}
		\end{center}
	\end{figure}
	
	\subsection{Quantum circuit for perfect storage and retrieval of $2$ unitaries}
	\label{sec:qcircuit}
	
	Next, we return to a qubit case ($d=2$) and we derive a simple quantum circuit realizing the optimal perfect probabilistic retrieval. We will use the following notation to simplify the formulas
	\begin{align}
\begin{aligned}
    \tilde{c}_n&:=\cos{2n\alpha} & \tilde{s}_n&:=\sin{2n\alpha} \\
	c_n&:=\cos{n\alpha} & s_n&:=\sin{n\alpha} \\
	\tilde{c}&:=\cos{2\alpha} & \tilde{s}&:=\sin{2\alpha}\\
	c & := \cos \alpha & s &:= \sin \alpha.
\end{aligned}
	\end{align}

	\begin{proposition}
		\label{prop:qc_for_psr}
		Quantum circuit depicted in Figure~\ref{fig:fig2} performs optimal perfect probabilistic retrieval of qubit unitaries $U_0$, $U_1$ if the qutrit-basis measurement outcomes $0,1$ are considered as success and $\sigma_Z$ correction is performed in case of outcome $1$. Measurement outcome $2$ is considered as a failure.
		In particular, we have
		\begin{align}
		\label{eq:psrworks}
		&{}_A\bra{0}\otimes I_2\;(M \otimes I)\, C_X
		\;\ket{\psi_{n,k}}\otimes I_1=
		\sqrt{\lambda_A} U_k \\
		&{}_A\bra{1}\otimes I_2\;(M \otimes \sigma_z)\, C_X
		\;\ket{\psi_{n,k}}\otimes I_1=
		\mp i
		\sqrt{\lambda_B} U_k\; , \nonumber
		\end{align}    
		where 
		\begin{align}
		\lambda_A&=\frac{1+\tilde{c}_n(\tilde{c}-\tilde{s})}{2} \quad \;
		\lambda_B=\frac{1-\tilde{c}_n(\tilde{c}+\tilde{s})}{2}, \nonumber
		\end{align}	
		and operator $M$ describes the following qubit to qutrit isometry
        \begin{align}
		\label{eq:def_gate_M}
		M :=
		\left \{
		\begin{aligned}
		&\ketbra{a_1}{+}+\ketbra{a_2}{-}
		&& \mbox{ if } \alpha \in [\chi_n , \frac{\pi}{4n}]\\
		& \ketbra{a'_1}{+}+\ketbra{a'_2}{-}
		&& \mbox{ if }
		\alpha \in (0,\chi_n)
		\end{aligned}  \right. ,
		\end{align} 
		with
		\begin{align}
		\ket{a_1}&:=\sqrt{\lambda_A}\frac{c}{c_n}\ket{0}+\sqrt{\lambda_B}\frac{s}{c_n}\ket{1}+ \sqrt{\nu_+}\ket{2} \nonumber \\
		\ket{a_2}&:=\sqrt{\lambda_A}\frac{s}{s_n}\ket{0}-\sqrt{\lambda_B}\frac{c}{s_n}\ket{1}- \sqrt{\nu_-}\ket{2} \nonumber \\
		\ket{a'_1}&:=\sqrt{a}\ket{0}+\sqrt{b}\ket{2} \quad
		\ket{a'_2}:=\sqrt{b}\ket{0}-\sqrt{a}\ket{2} \nonumber \\
		\nu_+&:=\frac{\tilde{c}_n (1-\tilde{c}^2+\tilde{s})}{1+\tilde{c}_n} \quad
		\nu_-:=\frac{\tilde{c}_n (\tilde{c}^2-1+\tilde{s})}{1-\tilde{c}_n} \nonumber\\
		a&:=\frac{(1+\tilde{c})(1-\tilde{c}_n)}{2(1-\tilde{c}\tilde{c}_n)} \quad
		b:=\frac{(1-\tilde{c})(1+\tilde{c}_n)}{2(1-\tilde{c}\tilde{c}_n)}. \nonumber
		\end{align}	
		
	\end{proposition}

	It is easy to verify that
	Equation~(\ref{eq:psrworks}) holds
	if we realize that 
	$U_0=c I+i s \sigma_z$, $U_1=c I-i s \sigma_z$ and write the Controlled-Not as
	$C_X=\ketbra{0}{0}\otimes I+\ketbra{1}{1}\otimes(\ketbra{+}{+}-\ketbra{-}{-})$. 
	The full proof of Proposition \ref{prop:qc_for_psr} and the derivation of the optimal circuit is placed in appendix \ref{sec:real-optim-perf}.
	
	\begin{figure}[t]
		\begin{center} \includegraphics[width=\columnwidth
			]{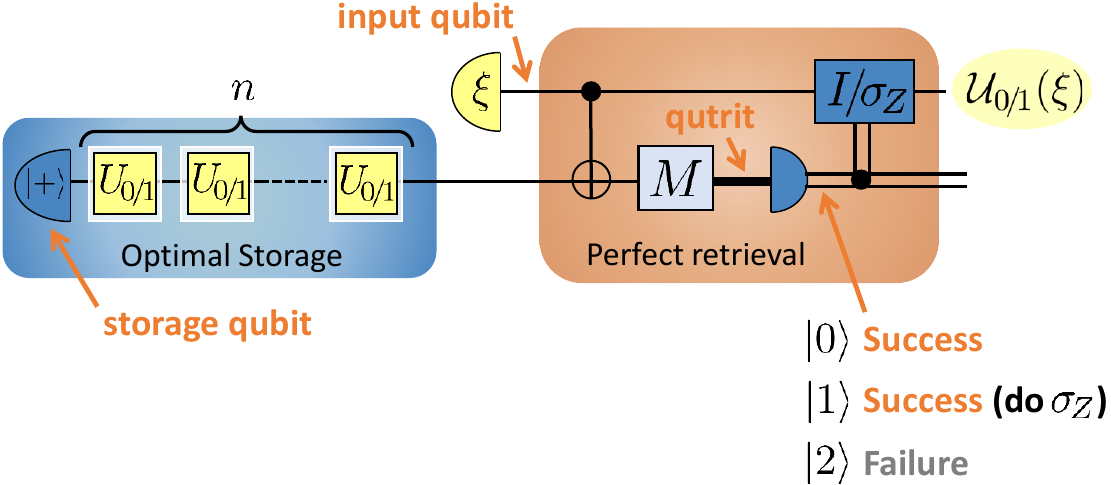} 
			\caption{Quantum circuit for optimal perfect probabilistic storage and retrieval
				of two single-qubit unitary transformations. Isometry $M$
				maps qubit to qutrit and together with the
				projective measurement of the qutrit forms a three 
				outcome rank 
				one 
				POVM.}
			\label{fig:fig2}
		\end{center}
	\end{figure}
	
	An appealing feature of the proposed circuit is its relatively small complexity - it contains only few elementary elements such as qubits, gates or their conditioning, and, most notably, only one entangling CNOT gate. One might wonder how close to ideal implementation he can get, especially from the point of view of perfectly implementing the stored unitary transformation.
	We  answer this question in case of quantum optical setups, by conducting the
	experiment reported in the following section.

	\section{Photonic demonstration}
	\label{sec:phot-demonstr}
	
	We present a proof of principle experiment that tests
	the probabilistic version of the proposed
	storage and retrieval 
	protocol described in Section \ref{sec:qcircuit} and depicted in Figure \ref{fig:fig2}. The experiment is based on a linear
	optical implementation \cite{Knill:2001aa}. Correlated photon pairs are
	generated in a CW-pumped type-II SPDC process at a
	central wavelength of 810 nm and guided via
	optical fibers to the experimental setup depicted
	in Figure~\ref{fig:fig3x}.
	
	\begin{figure}[b]
		\begin{center}
			\includegraphics[width=0.95\linewidth ]{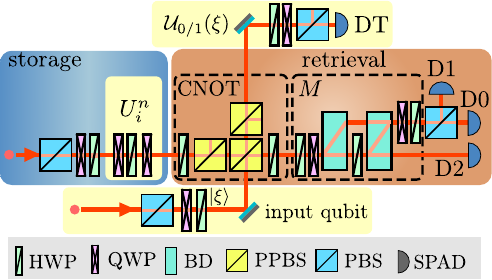}
			\caption{Experimental implementation of the probabilistic version of storage and retrieval protocol. Experimental setup consisting of half- and quarter-wave plates (HWP, QWP), (partially) polarizing beam-splitters ((P)PBS), calcite beam displacers (BD), and single-photon avalanche detectors (SPAD). The encoding of stored operation ($U_{i}^{n}$) is highlighted with a yellow box. CNOT gate and isometry $M$ are highlighted by dashed boxes.}
			\label{fig:fig3x}
		\end{center}
	\end{figure}
	
	The qubits are encoded into the polarization
	states of generated photons. One qubit serves for
	storage and is initially prepared in state
	$|+\rangle$ using wave plates. The subsequent
	stack of wave plates realizes operation
	$U_{i}^{n}$, introducing the phase shift to the
	initial storage state. Such a qubit could be, in
	principle, stored in photonic quantum memory \cite{PhysRevLett.108.190505,PhysRevLett.108.190504}.
	
	The CNOT gate and isometry $M$ followed by
	measurement in the qutrit computational basis
	implement the retrieval. To retrieve the stored
	operation and apply it onto the input qubit in a state $\ket{\xi}$,
	we first entangle both qubits using a
	probabilistic CNOT gate consisting of partially
	polarizing beam splitters (PPBS) and utilizing the
	two-photon 
	interference~\cite{PhysRevLett.95.210504,PhysRevLett.95.210505,PhysRevLett.95.210506}. The
	storage qubit enters the CNOT gate and acts as a
	target qubit. State $\ket{\xi}$ of the input qubit serves as a control and it 
	is prepared using another pair of wave plates.  The tomographic
	characterization showed CNOT gate process fidelity
	of 0.929(1) with the 1.8 kHz rate of detected
	coincidences.

	We realize the POVM measurement on the storage qubit
	using an optical setup for unambiguous state
	discrimination (USD). The optical setup for the
	USD consists of a Mach-Zehnder-type interferometer
	made with polarizing beam splitters and wave
	plates in its arms~\cite{PhysRevA.63.040305}. Such a
	device attenuates one of the polarization
	components, making the to-be-discriminated states
	orthogonal at the expense of success
	probability. We implement this device using wave
	plates and an interferometer based on calcite beam
	displacers, offering compact construction and
	passive phase stability~\cite{Starek:18}. The
	corresponding wave plates are set according to the
	choice of the states that are to be discriminated,
	which is parameterized by $|\alpha|$. The detailed
	description of this experimental block, including
	the relation between parameter $|\alpha|$ and wave
	plate angles, is provided in the 
	appendix \ref{app:isometry_par}.
	The success of the
	protocol is heralded by the detection of
	coincidence DT-D0 (outcome 0) or DT-D1 (outcome 1), while the coincidence
	DT-D2 (outcome 2) heralds the failure.

	In the case of coincidence detection at DT-D1, one
	has to apply an additional phase flip on the
	input 
	qubit. Such a feed-forward operation
	could be physically applied using a fast
	electro-optical phase modulator, fast digital
	logic, and a sufficiently long delay line for the 
	input qubit~\cite{Prevedel:2007aa,PhysRevLett.100.160502,Ma:2012aa,PhysRevLett.112.103602,PhysRevApplied.8.014016,Agresti:2020aa,LuizZanin:21,PhysRevLett.129.150501}. 
	However, in the tomographic characterization of the protocol, we can emulate the feed-forward using data post-processing~\cite{Horova:2022aa}. 
	The final state of the input qubit 
	is analyzed by means of projective measurements implemented by the remaining pair of wave plates, subsequent polarizing beam-splitter cube, and detector DT. The retrieved quantum operations are reconstructed using the maximum-likelihood method~\cite{PhysRevA.63.020101} from the collected tomographic data.

	Let us discuss in detail how we test the
	protocol. We sweep parameter $|\alpha|$ from $0$
	to $\frac{\pi}{4n}$, and for each value, we set
	the discrimination part accordingly. For
	$|\alpha| > \frac{\pi}{4n}$ the success
	probability is equal to one. We have tested both stored operations $U_{i}^n$, i.e., both positive and negative phase shift applied to the storage qubit.
	
	Having chosen 
	$U_i$ ($i=0,1$) in the storage phase we characterize
	the effective quantum channel  
	for the 
	input qubit using
	quantum process tomography. 
	We do that by sequentially preparing the input qubit in eigenstates of Pauli operators and measuring the corresponding output states in $X$, $Y$ and $Z$ basis. Here, $i$-basis measurement means the measurement of expectation value $\langle \sigma_i \rangle$. We record three
	coincidence tomograms, DT-D0 and DT-D1,
	corresponding to the protocol success, and DT-D2,
	corresponding to the failure. We flip the sign of
	$X$ and $Y$ basis readings in the DT-D1 tomogram
	to emulate the feed-forwarded conditional phase flip corresponding to the conditional $\sigma_0/\sigma_Z$ gate from Figure \ref{fig:fig2}.
	
	We can directly evaluate success probability from the tomograms as 
	\begin{align}
	\label{eq:defPtom}
	P_{succ} = \frac{1}{2}\left(\frac{S_{0+}+S_{1+}}{S_{0+}+S_{1+}+S_{2+}} + \frac{S_{0-}+S_{1-}}{S_{0-}+S_{1-}+S_{2-}}\right),
	\end{align}
	where $S_{j\pm}$ is the sum of all counts in the
	tomogram obtained from coincidence DT-D$j$ when
	the sign of the phase was +(-). This determination
	of the retrieval success probability is
	independent of the success probability of the CNOT
	gate. By summing tomograms DT-D0 and DT-D1, we
	obtain an effective tomogram, 
	which we reconstruct
	to get an estimate of the Choi-Jamiolkowski matrix $C_{exp,i}$ of the quantum channel acting on the 
	input qubit \emph{conditioned on the success of the protocol} (observation of DT-D0 or DT-D1 coincidence).
	
	We
	quantify the quality of the retrieval  with the average channel fidelity between the experimentally retrieved transformation and 
	the stored unitaries:
	\begin{align}
	\label{eq:defFidchi}
	F_{exp} = \frac{1}{8}\sum_{i=0,1}
	\Bra{U_i}C_{exp,i}\Ket{U_i}.
	\end{align}
	Thus, the probability of success and 
        channel fidelity $F_{exp}$, plotted using the blue color in Figure~\ref{fig:fig4}, are averages over the choice of $U_0$ or $U_1$ (sign of $\alpha$).
	
	\begin{figure}[t]
		\begin{center}
			\includegraphics[width=\linewidth ]{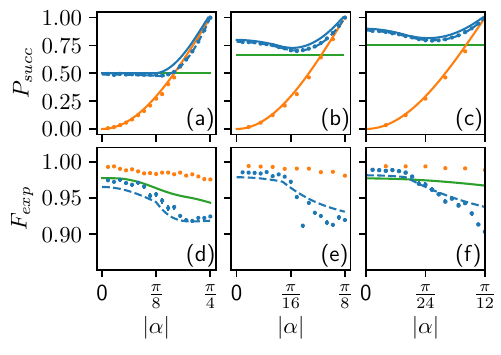}
			\caption{Success probability (a-c), and conditional process
				fidelity (d-f) for $n=1$ (a,d), $n=2$ (b,e), and
				$n=3$ (c,f) 
                uses 
                of the stored
				operation. Points show the experimental data,
				while solid lines represent the ideal theoretical
				predictions. The 1-$\sigma$ error bars are
				smaller than the point size. The blue color is
				related to the presented storage-and-retrieval
				protocol, and the orange to the
				measure-and-prepare strategy. The blue dashed
				lines represent the theoretical
				predictions based on the reconstructed process
				matrix of the CNOT gate and additional experimental imperfections. For comparison, green lines in panels (a,b,c,d,f) show the theoretical prediction for the storage and retrieval protocol from  \cite{PhysRevA.102.032618} that works for an unknown arbitrary phase $\alpha$ 
                and we assumed the same quality of CNOT gate. 
                The same green line for conditional process fidelity for $n=2$ is below the range reported here, since the 
                best corresponding
                quantum circuit 
                requires $8$ CNOT gates,
                in contrast to one or two CNOT gates needed in the other cases.
                }
			\label{fig:fig4}
		\end{center}
	\end{figure}
	
	For comparison, we also consider the
	perfect probabilistic measure-and-prepare strategy in which the storage
	qubit undergoes USD, and the result is stored in a
	classical bit if the USD succeeds. Upon retrieval,
	we simply apply $U_0$ or $U_1$ to the input qubit in state $\ket{\xi}$
	based on the value of the saved bit.  To test this
	in the experiment, we first prepared
	$U_0^n |+\rangle$ at the USD input, noting the
	relative frequencies of each outcome od the USD, which we denote as
	$f^{(0)}_{0}, f^{(0)}_{ 1}$, and $f^{(0)}_{ 2}$ (the outcome $2$ is an inconclusive outcome). 
	Theoretically, 
	$f^{(0)}_{1} = 0$, 
	when we store operation $U_{0}$, but
	due to experimental imperfections, there is a
	small fraction of wrong discrimination
	results. The success probability for stored $U_0$  is
	denoted as $P_{succ}^{(m\&p)}(U_0)$ and is estimated as follows
	\begin{align}
	\label{eq:defPtommp}
	P_{succ}^{(m\&p)}(U_0) = \frac{f^{(0)}_{0}+f^{(0)}_{1}}{f^{(0)}_{0}+f^{(0)}_{1}+f^{(0)}_{2}}.
	\end{align}
	We aim at evaluating the performance of the whole measure$\&$prepare strategy, thus we 
	also performed process tomography of the direct implementation of channels $U_{0}$
	and $U_{1}$, 
        which should be applied based on the USD outcome.
	Obtained Choi-Jamiolkowski matrices are denoted as
	$C_{dir,0}$ and $C_{dir,1}$.
	
	Then, the Choi operator of the retrieved operation when the inserted unitary is $U_0$ reads as follows:
	\begin{align}	
	C^{(m\&p)}_{exp,0} = \frac{f^
		{(0)}_{ 0}C_{dir,0} + f^
		{(0)}_{1}C_{dir,1}}
	{f^
		{(0)}_{0} + f^
		{(0)}_{1}}.     	
	\end{align} 
	This analysis is also done for storage state $U_1^n |+\rangle$, and finally the obtained fidelities of reconstructed Choi matrices are averaged to obtain $F^{(m\&p)}_{exp}$ as in Equation~(\ref{eq:defFidchi}).
	We repeat this procedure for each tested $\alpha$.

	The data in panels (a-c) of Figure~\ref{fig:fig4}
	show a good agreement between the observed success
	probability and its theoretical prediction. 
	Error bars in panels (a-c) are obtained
	directly by error propagation assuming Poissonian
	statistics of the detected counts.
	
	As demonstrated, the presented protocol surpasses the
	measure-and-prepare strategy in terms of success
	probability. Especially for low $\alpha$, the
	success probability 
	given by Equation~(\ref{eq:16a}) 
	is much higher than the
	one achievable by the measure-and-prepare
	strategy. However, the fidelity achieved by our
	protocol in the experiment is lower than the one
	obtained in the case of the measure-and-prepare
	strategy, as indicated in panels (d-f). To determine the error bars in panels (d-f) we resort to
	bootstrapping, i.e. we assume the Poissonian
	statistics of the measured counts and use our
	knowledge of measurement operators and the
	reconstructed Choi matrix to generate synthetic
	data~\cite{Micuda2017}. These data correspond to the repeated
	execution of the experiment, which we again
	reconstruct to obtain the fidelity. Because the
	fidelity distribution is asymmetric close to the
	boundary value of its definition interval, we use
	the 0.159 and 0.841 quantiles as lower and upper
	error bars. This choice is equivalent to one
	standard deviation in the case of symmetric
	distributions.
	
	We use the reconstructed Choi matrix of the CNOT
	gate to predict the expected fidelity for $n =
	1$. Such a model explains the fidelity drop by 3$\%$ in the
	small $\alpha$ 
	regime. Therefore, the fidelity is reduced mainly
	due to the CNOT gate, which is very sensitive to
	the partial distinguishability of single
	photons. The Hong-Ou-Mandel dip visibility of the
	photons produced by the SPDC source was 95$\%$,
	while the state-of-the-art single-photon sources
	can achieve 99$\%$~\cite{Somaschi:2016aa},
	substantially increasing the achievable fidelity.

	However, in the large $\alpha$ regime when 
	a $3$ outcome POVM needs to be realized instead of a projective measurement
	on the storage qubit, the observed fidelity drops
	even lower than is predicted by the 
	reconstructed Choi matrix of the CNOT gate alone. We suspect that the additional loss
	of fidelity is caused partly by a misalignment of
	the phase in the Mach-Zehnder interferometer and
	partly by disagreement between actual and expected
	introduced phase shift $\alpha$. We added these
	factors into the mathematical model and observed
	that such a prediction better describes the
	observed trend, as the dashed line
	illustrates in Figure~\ref{fig:fig4}. These errors cause the collective
	state of two qubits to collapse to a slightly
	different state upon measurement of the
	qutrit. This type of error is negligible in
	integrated photonic chips, which provide
	interferometric-based single-qubit gates with
	fidelities
	reaching~99.9$\%$~\cite{doi:10.1126/science.aab3642}. The
	remaining discrepancy between the model and the
	data is caused by neglecting other imperfections,
	such as deviations of wave plate retardances from their nominal
	values, path-dependent losses, multi-photon contributions in our photon pair source and various instabilities in the setup.
	
	On the other hand, the measure-and-prepare
	strategy requires only single-qubit operations,
	which are very robust for polarization
	qubits.
	
	We also compare the realized protocol to the protocol for storing arbitrary phase gates~\cite{PhysRevA.102.032618} in the case of $n=1$ and $n=3$. In the compared protocol, one first 
    applies the stored phase shift once to the \emph{storage qubit 1}. If 
    we are allowed to use $U_i$ three times, we also apply $U_i$ twice to the \emph{storage qubit 2}, 
    which ends the storage phase.    
    Then in the retrieval phase, we apply the CNOT gate controlled by \emph{input qubit} to the \emph{storage qubit 1}. If the subsequent measurement of \emph{storage qubit 1} in the Z basis yields 0, then the protocol succeeded. If it yields 1 
    it implies failure for $n=1$, but for $n=3$ 
    we apply another CNOT gate controlled by the \emph{input qubit} to the \emph{storage qubit 2} and measure it in the Z basis. 
    Such a protocol would succeed independently on $|\alpha|$ with probability $P_{succ} = \frac{1}{2}$ for $n = 1$ and $P_{succ} = \frac{3}{4}$ for $n = 3$, respectively. In Figure~\ref{fig:fig4}, we see that the presented protocol described by Proposition \ref{prop:qc_for_psr} (blue data points) already outperforms the protocol~\cite{PhysRevA.102.032618} (green lines) in terms of success probability.
	
	We utilized the reconstructed Choi matrix of the CNOT gate to theoretically compare both protocols in terms of their robustness to realistic experimental imperfections. We input the imperfect experimental Choi matrix into numerical simulations of the protocol \cite{PhysRevA.102.032618} and plotted the achieved fidelity as green solid lines in panels (d) and (f) of Figure~\ref{fig:fig4}. For $n=1$ the ideal versions of the compared protocols coincide, i.e. demand the same quantum circuit to be performed, in the small $\alpha$ regime. From that part of Figure~\ref{fig:fig4} we see that, the experimentally realized protocol 
    has slightly smaller process fidelity due to imperfections in the performed projective measurement, which are not accounted for in the predictions illustrated by the green line.  
    In the large $\alpha$ regime presented protocol has higher success probability, but 
    due to additional disturbance caused by the imperfect POVM realization lower process fidelity. However, for $n=3$, the protocol~\cite{PhysRevA.102.032618} can use up to two applications of the CNOT gates, which 
    on average results in lower fidelity, as indicated in panel (f). Therefore, the presented protocol is more robust to experimental imperfections and at the same time provides higher success probability in this case.

	\section{Discussion}
	\label{sec:discussion}

	We studied storage
	and retrieval of unknown unitary transformation,
	which is chosen from set of two known
	possibilities with equal prior probabilities and
	the retrieved transformation is expected to mimic
	well the action of the stored unitary on any input
	state. 
	We proved that the optimal state for storage is the one which is also optimal for the discrimination of the two unitaries and we optimized the retrieval in two scenarios. 
	If the goal is to maximize the average process fidelity for a deterministic strategy,  the  solution (for qubit unitaries) is of the kind "measure-and-prepare" (see Figure~\ref{fig:meas_prep}) and the optimal value of the average fidelity reads as in Equation~\eqref{eq:det_fom_proc}.
	It is easy to observe that this result is also optimal if we want to retrieve $m \geq 1$ copies of the unknown unitary with the maximal single-copy fidelity. 
	
	We also found the optimal probability of success in the case of a perfect probabilistic retrieval in arbitrary dimension.
	The result is given by Equation~\eqref{eq:fom_final}  and we observe that
	for
	small values of
	$\alpha$ the success probability  approaches unity 
	as  $1/n^2$ with a  quadratic improvement over
	known storage and retrieval strategies, which work
	for continuous subsets of transformations.

	The retrieval  phase can be also viewed as a
	programmable quantum processor, since the state
	resulting from the storage phase programs
	(encodes) the transformation to be performed on
	another system. In a suitable basis, the program
	state encodes a phase shift
	of $2n\alpha$ and the processor should induce a
	unitary transformation
	(or its inverse) with relative phase shifts of at
	most $2\alpha$. Since in our optimization of
	retrieval operation parameter $n$ can be treated
	as continuous we can think of this processor as
	optimally transforming program states with
	relative phase shift $\pm \beta$
	(i.e. states with overlap $\cos{\beta}$) into two
	unitary transformations $U_0,U_1$, whose relative
	unitary $U_1^\dag U_0$ has angular spread of eigenvalues at
	most $4\alpha$. Thus, by performing the
	optimization of the retrieval we found two classes
	of optimal processors, which for any given overlap
	of input states and any choice of two unitary
	transformations achieve the best possible average
	process fidelity or best average probability of
	success, respectively. 
	Since the angular spread is the only relevant parameter, our results easily generalize to the case in which, instead of retrieving $U_0$ or $U_1$, the goal of the protocol is to invert (i.e. retrieving $U_i^\dag$), conjugate (i.e. retrieving $U_i^*$ ) or transpose (i.e. retrieving $U_i^T$) the unknown unitary.

	For qubit unitary transformations we
	found a  quantum circuit performing the optimal
	perfect probabilistic storage and retrieval using just a
	single CNOT gate and a suitable unambiguous
	measurement on its target qubit (see Figure
	\ref{fig:fig2}). For quite distinct unitary
	transformations (large $\alpha$ regime) two of the
	measurement outcomes correspond to success, but
	for one of them unitary ($\sigma_Z$ gate)
	correction must be performed on the unmeasured
	qubit. For almost identical unitary
	transformations that should be stored (small
	$\alpha$ regime) the measurement is projective
	with just success (no correction) and failure
	results.  
	
	We built a linear optical 
	experiment implementing  the presented optimal
	circuit. The setup is based on the
	CNOT gate consisting of partially polarizing beam
	splitters (PPBS) and utilizing the two-photon
	interference~\cite{PhysRevLett.95.210504,PhysRevLett.95.210505,PhysRevLett.95.210506, PhysRevA.89.042304}.  The unambiguous measurement part 
	consists of a Mach-Zehnder-type interferometer
	made with polarizing beam splitters and wave
	plates in its arms. Tomographically complete data
	were recorded and reconstructed using the
	maximum-likelihood method. The setup was also used
	to test non-optimal measure-and-prepare strategy
	for the perfect retrieval. The success probability
	of both strategies quite closely follows the
	theoretically predicted curves. Full tomographic
	data allowed us also to calculate process fidelity
	of the recreated transformations. Naturally,
	measure-and-prepare strategy achieves better
	fidelities as it uses less multipartite
	operations, actually only those needed to perform
	the unambiguous discrimination
	measurement. Optimal perfect retrieval uses also
	the CNOT gate, which due to
	its 
	$93\%$ fidelity of realization is the main cause
	of the observed infidelities (especially in the
	small $\alpha$ regime). Perfect retrieval
	outperforms the measure-and-prepare strategy
	especially from the success probability point of
	view. For $n>1$ this holds for almost any pair of
	unitaries (i.e. the whole interval of $\alpha$
	parameter). For $n=1$ and small $\alpha$ regime
	the perfect retrieval coincides with probabilistic
	programmable implementation of phase gates
	proposed in Ref. \cite{PhysRevLett.88.047905}. On the other hand, 
	Ref. \cite{PhysRevA.102.032618}
	shows that two applications of
	this gate implement $n=3$ storage and
	retrieval of arbitrary qubit phase gate with
	probability $3/4$. 
	The	reported experimental setup outperforms such
	strategy both in success probability and in
	achieved process fidelity if we assume the same
	quality of CNOT gate as in our experiment.
	
	In the same way as understanding
	the tasks related to non-orthogonality of states  helps to understand several
	quantum mechanical phenomena \cite{Fuchs2002}, 
	the presented
	results 
	shed light  to the ultimate limits which quantum mechanics impose to the manipulation of transformations.  
	Future
	research could cover cloning of two unitary
	transformations and other tasks following the
	spirit of the above ideas.  Some technical questions are still open, such as
	optimization of approximate storage and retrieval
	for qudits or search for an efficient quantum circuit
	for storage and retrieval of qudit unitary
	transformations in the perfect probabilistic case.

	\acknowledgments 
	A.B. acknowledges financial support from the European Union - Next Generation E.U. through the PNRR MUR
Project PE0000023-NQSTI . R.S. and M.M. acknowledge the support from the Ministry of the Interior of the Czech Republic, project NU-CRYPT (VK01030193), and from Ministry of Education, Youth and Sports of the Czech Republic, grant no. 8C22003 (QD-E-QKD) of the QuantERA II Programme that has received funding from the European Union’s Horizon 2020 research and innovation programme under Grant Agreement no. 101017733. J.F. acknowledges the support from Palacky University, IGA-PrF-2024-008. M.S. was supported by projects APVV-22-0570 (DeQHOST), VEGA 2/0183/21 (DESCOM). M.S. was further supported by funding from QuantERA, an ERA-Net cofund in Quantum Technologies, under the project eDICT.

	\section*{Author contributions}
	A.B. and M.S. conceptualized the problem and performed the theoretical calculations.
	J.F. contributed to the theoretical derivation and supervised the experimental work. 
	R.S., N.H., and M.M. constructed the experimental setup, conducted the measurements, and analyzed the acquired data. All authors contributed to the manuscript's writing.

	\appendix
	
	\section{Proof of Proposition \ref{prp:optimal-approxdet_retrieval_compactform}}
	\label{sec:appqrsq}
	We will split the proof of Proposition
	\ref{prp:optimal-approxdet_retrieval_compactform}
	into several steps.
	
	Reminding the decomposition $R = PRP + P'RP'$ of
	Lemma \ref{lmm:blockdiagretriev}, we can write:
	\begin{align}
	\label{eq:blockdiagR}
	\begin{aligned}
	PRP &= \sum_{i,j=1}^2A_{i,j} \KetBra{e_i}{e_j} \\
	P'RP' &=
	\sum_{i,j=1}^2B_{i,j}
	\KetBra{e'_i}{e'_j}
	\end{aligned}
	\end{align}

	\begin{lemma}
		\label{lmm:PT}
		$\Tr_2[R] \leq I$ if and only if
		$A + B \leq I $, where $A$ and $B$ are the
		$2\times 2$ matrices with coefficient $A_{i,j}$
		and $B_{i,j}$ introduced in Equation~\eqref{eq:blockdiagR}.
	\end{lemma}
	\begin{proof}
		A straightforward computation gives
		\begin{align*}
		&  \Tr_2[\KetBra{e_1}{e_1}] =
		\ketbra{+}{+}  \otimes I \quad
		\Tr_2[\KetBra{e_1}{e_2}] =
		\ketbra{+}{-} \otimes \sigma_z \\
		& \Tr_2[\KetBra{e_2}{e_1}] =
		\ketbra{-}{+}  \otimes \sigma_z \quad
		\Tr_2[\KetBra{e_2}{e_2}] =
		\ketbra{-}{-} \otimes I  \\
		&  \Tr_2[PRP] = \widetilde{A} \otimes \ketbra{0}{0} +
		\sigma_x \widetilde{A}\sigma_x\otimes \ketbra{1}{1} \\
		& \Tr_2[P'RP'] = \widetilde{B} \otimes \ketbra{0}{0} +
		\sigma_x \widetilde{B} \sigma_x\otimes \ketbra{1}{1}
		\end{align*}
		where $\widetilde{X}= X_{1,1} \ketbra{+}{+}+X_{1,2} \ketbra{+}{-}+ X_{2,1} \ketbra{-}{+}+X_{2,2}\ketbra{-}{-}$.
		Then we have
		\begin{align*}
		\Tr_2[R] = 
		(\widetilde{A}+\widetilde{B}) \otimes \ketbra{0}{0} +
		\sigma_x(\widetilde{A}+\widetilde{B})\sigma_x \otimes\ketbra{1}{1}.
		\end{align*}
		and we have that $\Tr_2[R] \leq I_{01}$ if and only if
		$\widetilde{A}+\widetilde{B} \leq I $, which is equivalent with $A+B \leq I$.
	\end{proof}
	
	\begin{lemma}
		\label{lmm:figmeritwithAandB}
		Figure of merit from Equation~(\ref{eq:optprob1}) can be expressed as
		$  \Tr[RD] = \bra{u}A\ket{u} + \bra{v}B\ket{v}$.
	\end{lemma}
	\begin{proof}
		It is easy to verify that
		\begin{align}
		\label{eq:14}
		&  \begin{aligned}
		&
		\ket{\psi^*_{n,0}} =  c_n\ket{+} -i s_n \ket{-},
		& \Ket{U_0} &= c\Ket{I} + i s\Ket{\sigma_z}
		,\\
		&
		\ket{\psi^*_{n,1}} = \sigma_x   \ket{\psi^*_{n,0}},
		& \Ket{U_1} &= \sigma_x\otimes\sigma_x \Ket{U_0}  ,
		\end{aligned}
		\end{align}
		which implies that
		\begin{align}
		\label{eq:19}
		\begin{aligned}
		& \begin{aligned}
		\ket{\psi^*_{n,0}} \Ket{U_0}
		=&
		\Big( c_nc \Ket{e_1} + s_ns
		\Ket{e_2} \Big)+\\
		&+ i \Big( c_ns \Ket{e'_1} - s_nc \Ket{e'_2} \Big )
		\end{aligned} \\
		&\ket{\psi^*_{n,1}} \Ket{U_1} = \sigma_x^{\otimes 3}\ket{\psi^*_{n,0}} \Ket{U_0}
		\end{aligned}
		\end{align}
		By substituting Equation~\eqref{eq:19}
		into Equation~\eqref{eq:5} and using Equation~\eqref{eq:blockdiagR} we obtain the thesis.
	\end{proof}
	
	\begin{lemma}
		\label{lmm:perfretrievconditionwithAandB}
		The perfect retrieval condition of Equation~\eqref{eq:1} is equivalent to
		$\bra{v}A \ket{v} = \bra{u}B \ket{u}=0$.
	\end{lemma}
	\begin{proof}
		Via direct computation using Eqs. (\ref{eq:defee}), (\ref{eq:blockdiagR}), (\ref{eq:14}) we can rewrite Equations~\eqref{eq:1} 
		as
		\begin{align}
		\label{eq:fformAB}
		&   A = \lambda_A \ketbra{\phi_A}{\phi_A} \quad  B = \lambda_B \ketbra{\phi_B}{\phi_B} \\
		&   \ket{\phi_A} :=
		\begin{pmatrix}
		\frac{c}{c_n} \\
		\frac{s}{s_n}
		\end{pmatrix}
		\quad
		\ket{\phi_B} :=
		\begin{pmatrix}
		\frac{s}{c_n} \\
		-\frac{c}{s_n}
		\end{pmatrix}
		\nonumber
		\end{align}
		We notice that $\braket{u}{\phi_B}=\braket{v}{\phi_A}=0$ and since we are in a two-dimensional space,  we have the thesis
	\end{proof}
	Combining Lemma  \ref{lmm:PT} with Lemma \ref{lmm:figmeritwithAandB} yields to Equation~\eqref{eq:optprob_approxdet_final}.
	Combining Lemma  \ref{lmm:PT}, Lemma \ref{lmm:figmeritwithAandB} and Lemma \ref{lmm:perfretrievconditionwithAandB} yields to Equation~\eqref{eq:optprob_perfprob_final}.

	\section{Realization of the optimal approximate deterministic retrieval }
	\label{sec:real-optim-appr}
	The optimal operators $A$, $B$ from Equation~(\ref{eq:optprob_approxdet_final}) are one-dimensional
	orthogonal projectors defining a projective
	measurement in basis
	\begin{align}
	\label{eq:meformAB}
	&   A = \ketbra{\phi_A}{\phi_A} \quad  B = \ketbra{\phi_B}{\phi_B} \\
	&   \ket{\phi_A} :=
	\begin{pmatrix}
	a \\
	b
	\end{pmatrix}
	\quad
	\ket{\phi_B} :=
	\begin{pmatrix}
	-b^* \\
	a*
	\end{pmatrix} ,
	\nonumber
	\end{align}
	where $|a|^2+|b|^2=1$ and $a,b$ are suitable real
	numbers. 
	Note that the figure of merit (\ref{eq:optprob_approxdet_final}) can be written as 
	$F_e = \eta_u \bra{\tilde{u}}A\ket{\tilde{u}} + \eta_v\bra{\tilde{v}}B\ket{\tilde{v}}$,
	where $\ket{\tilde{u}} :=\frac{1}{\sqrt{\eta_u}}\ket{u}$, 
	$\ket{\tilde{v}} :=\frac{1}{\sqrt{\eta_v}}\ket{v}$ are normalized pure states and $\eta_u+\eta_v=\braket{u}{u}+\braket{v}{v}=1$. 
	Due to discussion below Equation~(\ref{eq:optprob1}) we know $A+B=I$. Consequently, the problem is actually a minimum error discrimination problem for pure states $\ket{\tilde{u}}, \ket{\tilde{v}}$ appearing with prior probability $\eta_u,\eta_v$, respectively.
	Then,  the parameters $a$ and $b$ can be determined by the solution of this state discrimination problem
	but
	their particular values are not important for the
	upcoming discussion.
	However, to show that
	realization scheme based on vectors
	$\ket{\phi_A}$, $\ket{\phi_B}$ have the same
	optimal performance it is useful to rewrite figure
	of merit in terms of parameters $a,b$,
	which reads as follows
	\begin{align}
	\label{eq:optmepval}
	F &= (a\,c_n c + b\,s_n s)^2+ (a\,s_n c+b\, c_n s)^2.
	\end{align}
	
	The matrices $A,B$ completely determine the
	Choi-Jamiolkovski operator of $\mathcal{R}$, which
	has rank two. Thus, any dilation of $\mathcal{R}$
	must have at least two-dimensional ancila.  One of such
	minimal dilations is defined by the following
	unitary transformation
	\begin{align}
	\label{eq:defGmin}
	G
	=&(\ketbra{a_1}{+}+\ketbra{a_2}{-})\otimes \ketbra{0}{0} \nonumber\\
	&+(\ketbra{b_1}{+}+\ketbra{b_2}{-})\otimes \ketbra{1}{1}
	\end{align}
	where
	\begin{align}
	\ket{a_1}&=a\ket{0}-b^*\ket{1} \qquad 
	\ket{a_2}=b\ket{0}+a^*\ket{1} \nonumber \\
	\ket{b_1}&=\sigma_z\ket{a_1} \qquad
	\ket{b_2}=-\sigma_z\ket{a_2}.
	\end{align}
	We can write $G$ as
	\begin{align}
	G&=C_{\pi} \, (M \otimes I )\, \CNOT
	\end{align}
	where $\CNOT=I \otimes \ketbra{0}{0}+ \sigma_x \otimes \ketbra{1}{1}$ is the controlled-NOT gate,
	$M=\ketbra{a_1}{+}+\ketbra{a_2}{-}$ is qubit
	unitary transformation and
	$C_{\pi}=I \otimes \ketbra{0}{0}+ \sigma_z \otimes
	\ketbra{1}{1}$ is a controlled phase gate. The actual effect of $C_{\pi}$ when it is followed
	by the measurement in the basis
	$\{\ket{0}_A,\ket{1}_A\}$ of the first ancillary
	qubit can be equivalently achieved by omitting $C_{\pi}$ and taking no action or application of $\sigma_z$ gate on the second qubit if the
	measurement outcome was $0$ or $1$, respectively.
	
	In order to prove optimality of the proposed scheme, we first express the conditionally prepared unitary transformations
	\begin{align}
	\bra{0}\otimes I \,G\,\ket{\psi_{in,0}}\otimes I &=a\, c_n I + i b\,s_n \sigma_z =: \sqrt{p}\; U_{c,00} \nonumber \\
	\bra{1}\otimes I \,G\,\ket{\psi_{in,0}}\otimes I &=i a^*\, s_n - b^*\,c_n \sigma_z = :\sqrt{1-p}\; U_{c,01} \nonumber \\
	\bra{0}\otimes I \,G\,\ket{\psi_{in,1}}\otimes I &=a\, c_n I - i b\,s_n \sigma_z = :\sqrt{p} \;U_{c,10} \nonumber \\
	\bra{1}\otimes I \,G\,\ket{\psi_{in,1}}\otimes I &=-i a^*\, s_n - b^*\,c_n \sigma_z = :\sqrt{1-p} \; U_{c,11} , 
	\end{align}
	where $U_{c, ab}$ are qubit unitary transformations and $p=a^2 c_n^2 + b^2 s_n^2$. As a consequence the Choi operators of the prepared channels in case of the memory states $\ket{\psi_{n,0}}$, $\ket{\psi_{n,1}}$ are
	\begin{align}
	X_0=&p \KetBra{U_{c,00}}{U_{c,00}}+(1-p)\KetBra{U_{c,01}}{U_{c,01}} \nonumber \\
	X_1=&p \KetBra{U_{c,10}}{U_{c,10}}+(1-p)\KetBra{U_{c,11}}{U_{c,11}}
	,
	\end{align}
	respectively.
	Inserting the above expressions into the formula for the average process fidelity of the storage and retrieval we find that
	\begin{align}
	F_e = & \frac{1}{8}(\Bra{U_0} X_0 \Ket{U_0} + \Bra{U_1} X_1 \Ket{U_1}) \nonumber \\
	= & (a\,c_n c + b\,s_n s)^2+ (a\,s_n c+b\, c_n s)^2,
	\end{align}
	which coincides with the optimal value of the figure of merit (see Equation~(\ref{eq:optmepval})). 
	On the other hand,
	unitary transformation $G$, defined in Equation~(\ref{eq:defGmin}), can be rewritten also in the following way
	\begin{align}
	\begin{aligned}
	G=&\frac{1}{\sqrt{2}}(\ket{0}+i\ket{1})\bra{\uparrow} \otimes U_{+A} + \\
	&+ \frac{1}{\sqrt{2}}(\ket{0}-i\ket{1})\bra{\downarrow} \otimes U_{-A},
	\end{aligned}
	\end{align}
	where $U_{\pm A}=a I \pm i b\,\sigma_z$ are qubit
	unitary gates and
	$\ket{\uparrow}=\frac{1}{\sqrt{2}}(\ket{+}+i\ket{-})$,
	$\ket{\downarrow}=\frac{1}{\sqrt{2}}(\ket{+}-i\ket{-})$. Since
	$G$ is used for a dilation of a channel, it does
	not matter in which basis we perform the partial
	trace over the first qubit. Operationally it is
	equivalent to measuring in some basis (say
	$\frac{1}{\sqrt{2}}(\ket{0}\pm i\ket{1})$) and
	ignoring the measurement outcome. With this view
	we realize that application of $G$ and measurement
	in basis $\frac{1}{\sqrt{2}}(\ket{0}\pm i\ket{1})$
	can be seen as a realization of the following
	measure-and-prepare strategy for retrieval stage
	of the protocol. The input state
	$\ket{\psi_{n,0}}$ or $\ket{\psi_{n,1}}$ is
	first measured in the basis
	$\{\ket{\uparrow}, \ket{\downarrow} \}$ and
	subsequently unitary transformation $U_{+A}$ or
	$U_{-A}$ is applied to the main system depending
	on the obtained outcome. Of course, the
	measurement could be already performed at the end
	of the storage phase of the protocol, thus we
	found that for maximal process fidelity of
	deterministic storage and retrieval there exists a
	measure-and-prepare strategy (see illustration in Figure \ref{fig:meas_prep}) which
	does not require quantum memory and performs
	optimally.

	\section{Perfect probabilistic retrieval for $d > 2$}
	\label{sec:isom-appr-perf}

	In this section we will prove that the value of
	the probability of success derived in Equation~\eqref{eq:fom_final} for the qubit case is also
	optimal for generic dimension $d$.
	
	Let us denote with $P_{succ}(U_0,U_1)$
	the optimal probability
	of success for the perfect probabilistic storage
	and retrieval of the pair $\{U_0,U_1\}$
	of unitaries in dimension
	$d$.
	
	Without loss of generality, let $U_0$ and $U_1$
	be given as in Equation~\eqref{eq:26}.
	Let us consider the unitaries
	\begin{align}
	\label{eq:12}
	\begin{aligned}
	&\tilde{U}_0   :=
	e^{ i \alpha} \ketbra{0}{0} +
	e^{- i \alpha} \ketbra{d-1}{d-1}  \\
	&  \tilde{U}_1   :=
	e^{ -i \alpha} \ketbra{0}{0} +
	e^{ i \alpha} \ketbra{d-1}{d-1}.
	\end{aligned}
	\end{align}

	Qubit unitaries $\tilde{U}_0  $ and $\tilde{U}_1$
	are obtained by restricting the action of $U_0$ and $U_1$
	on the subspace spanned by $\ket{0}$ and $\ket{d-1}$.
	Only this subspace is used for the optimal storage (see Proposition \ref{prop:optimalstorage}) in both considered situations ($U_0$ or $U_1$ vs. $\tilde{U}_0$ or $\tilde{U}_1$), thus the same two states $\ket{\psi_{n,i}}$ $i=0,1$ are obtained after the optimal storage. Moreover, the retrieval for $U_0/U_1$ can be easily modified to perform retrieval for $\tilde{U}_0/\tilde{U}_1$, thus we conclude $P_{succ}(U_0,U_1) \leq
	P_{succ}(\tilde{U}_0,\tilde{U}_1) = P_{succ}$,
	where $P_{succ}$ is given by Equation~\eqref{eq:fom_final}. 

	We will now present a retrieval strategy that
	shows that $P_{succ}(U_0,U_1)$ achieves the upper bound
	given by the qubit case.
	The retrieval strategy is realized by composing
	an isometry
	$G : \mathcal{H}_{1} \otimes \mathcal{H}_{0}
	\mapsto \mathcal{H}_{2} \otimes \mathcal{H}_{3}
	\otimes \mathcal{H}_4$ and a projective
	measurement
	\begin{align}
	\label{eq:11}
	\Pi_s = I \otimes \ketbra{0}{0} \quad
	\Pi_f = I \otimes \ketbra{1}{1}
	\end{align}
	on system $\mathcal{H}_{3}\otimes \mathcal{H}_4 $
	($\mathrm{dim}(\mathcal{H}_4) = 2$),
	where $s/f$ is the
	outcome corresponding to a successful/failed
	retrieval.
	Let us consider the following ansatz for the isometry $G$:
	\begin{align}
	\label{eq:9iso}
	&\begin{aligned}
	G\ket{k}\ket{\psi_{n,i}} =&
	\sqrt{P_{succ}} \;e^{i\beta_{k,i}}\ket{k} \ket{\phi_i}\ket{0} +\\
	&+\sqrt{1-P_{succ}}\; \ket{k} \ket{\eta_{i,k}} \ket{1}
	\end{aligned}
	\\
	&\begin{aligned}
	& \ket{\phi_i}, \ket{\eta_{i,k}} \in \mathcal{H}_3,
	\quad \ket{0}, \ket{1}\in   \mathcal{H}_4,
	\\
	&\braket{\phi_i}{\phi_i} =
	\braket{\eta_{i,k}}{\eta_{i,k}} = 1 \\
	& \braket{\phi_0}{\phi_1} = x, \quad
	\braket{\eta_{0,k}}{\eta_{1,k'}} =
	y_{k,k'} \\
	&x, y_{k,k'} \in \mathbb{C} \quad |x|,|y_{k,k'}| \leq 1.
	\end{aligned}
	\end{align}
	where $\ket{k}$ is the basis of eingestates of
	$U_0$ and $U_1$, and we defined
	$\beta_{0,0} = \beta_{d-1,1} = \alpha$ ,
	$\beta_{d-1,0} = \beta_{0,1} = -\alpha$, 
	$\beta_{k,0} = \beta_{k} = - \beta_{k,1}$ in accordance with Equation~\eqref{eq:26}.
	Since
	$\{\ket{k}\ket{\psi_{n,i}}\}$ is a basis of
	$\mathcal{H}_1 \otimes \mathcal{H}_0$, $G$ is
	completely specified by Equation~\eqref{eq:defGmin}.
	It is clear that $G$ and the POVM
	$\{\Pi_s , \Pi_f\}$ provide a perfect
	probabilistic retrieval of the pair $\{U_0,U_1\}$
	with probability of success $P_{succ}$.  We only need to
	verify that $G$ is an isometry, and that is true
	if and only if all the scalar products are
	preserved. Clearly, $\ket{k}\ket{\psi_{n,i}}$ and
	$\ket{k'}\ket{\psi_{n,j}}$ are sent to orthogonal
	states for $k\neq k'$. Therefore, $G$ is an
	isometry if and only of the following set of
	equations is satisfied:
	\begin{align}
	\label{eq:13}
	& \cos(2n\alpha) = P_{succ}\; e^{i(\beta_{k,1} - \beta_{k,0})}x +
	(1-P_{succ})  y_{k,k} \quad \forall k \\
	&\mbox{where }|x|, | y_{k,k}| \leq 1 , \mbox{ and }  -2\alpha \leq \beta_{k,1} - \beta_{k,0} \leq 2 \alpha.
	\end{align}
	
	A rather lenghty but straightforward calculation
	shows that Equation~\eqref{eq:13} is solved by
	the following choice of parameters which
	achieve the same 
	probability of success as the qubit case:
	\begin{align}
	\label{eq:15}
	& \begin{aligned}
	&x = \frac{\tilde{c}_n \tilde{c}}{1- \tilde{c}_n \tilde{s}}, \;
	y_{k,k} = \frac {1-\tilde{c} \, \zeta_k}{\tilde{s}}, \;
	P_{succ} = 1- \tilde{c}_n \tilde{s}   \\
	& \mbox{if } \alpha \in [{\chi_n},\frac{\pi}{4n}]
	\end{aligned}
	\\
	&      \begin{aligned}
	&x = 1, \;
	y_{k,k} = \frac {\tilde{c}_n - P_{succ}\, \zeta_k}{1-P_{succ}}, \;
	P_{succ} = \frac{1-\tilde{c}_n^2}{2(1-\tilde{c}_n \tilde{c})}\\
	&\mbox{if } \alpha \in (0,{\chi_n})
	\end{aligned} \\
	&   \begin{aligned}
	&\zeta_k:= e^{i(\beta_{k,1} - \beta_{k,0})}, \quad \tilde{c}_n:=\cos(2n\alpha), \\
	&\tilde{c} := \cos(2\alpha), \quad
	\tilde{s} := \sin(2\alpha). 
	\end{aligned}
	\end{align}.

	\section{Realization of the optimal perfect probabilistic retrieval}
	\label{sec:real-optim-perf}

	The goal of this section is to show how the optimal quantum circuit was found and thus prove Proposition \ref{prop:qc_for_psr}. 
	We start from Equation~(\ref{eq:fformAB}). We notice that
	$\braket{u}{\phi_A}=\braket{v}{\phi_B}=1$ and consequently from Equation~(\ref{eq:optprob_perfprob_final}) we get
	\begin{align}
	P_{succ}= \lambda_A + \lambda_B. 
	\end{align}
	We rewrite perfect retrieval conditions from Equation~(\ref{eq:optprob_perfprob_final}) as
	\begin{align}
	\label{eq:1111}
	&   \bra{\widetilde{v}}A\ket{\widetilde{v}} = 0 \quad  \bra{\widetilde{u}}B\ket{\widetilde{u}} = 0 
	\end{align}
	where $ \ket{\widetilde{u}} :=\frac{1}{\sqrt{\eta_u}}\ket{u}$, $\ket{\widetilde{v}} :=\frac{1}{\sqrt{\eta_v}}\ket{v}$
	are normalized pure states.
	On the other hand, the figure of merit 
	from Equation~(\ref{eq:optprob_perfprob_final}) 
	reads
	\begin{align}
	\label{eg:FOM_seen_as_USD}
	&P_{succ}=
	\eta_u \bra{\widetilde{u}}A\ket{\widetilde{u}} + \eta_v\bra{\widetilde{v}}B\ket{\widetilde{v}},
	\end{align}
	where $\eta_u+\eta_v=\braket{u}{u}+\braket{v}{v}=1$. Due to normalization condition $A+B\leq I$ (Lemma \ref{lmm:PT}) operators $A,B$ can be interpreted as elements of a $3$-outcome POVM $\{A,B,I-A-B\}$, which performs unambiguous discrimination of pure states
	$\ket{\widetilde{u}}, \ket{\widetilde{v}}$ appearing with prior probability $\eta_u,\eta_v$, respectively. The figure of merit 
	(\ref{eg:FOM_seen_as_USD}) 
	equals the success probability of the discrimination and its optimal solution is known \cite{jaeger1995optimal}. 
	Depending on the prior probability $\eta_u$ and the overlap $\mu^2=|\braket{\widetilde{u}}{\widetilde{v}}|^2$ three regimes for the optimal measurement and success probability exist (see e.g. \cite{SedlakActaUnambigous} page $672$). However, for storage and retrieval of two unitaries only two of them will occur since $\eta_u=(1+\cos{2n\alpha}\cos{2\alpha})/2\geq 1/2$. In particular, we obtain Equation~(\ref{eq:fom_final}) expressed in suitable form for the upcoming steps
	\begin{align}
	&
	P_{succ}=
	\begin{cases}
	1-\cos{2n\alpha}\sin{2\alpha} & \eta_u \leq \frac{1}{1+\mu^2} \; (large \,\alpha)
	\\
	\frac{1-(\cos{2n\alpha})^2}{2(1-\cos{2n\alpha}\cos{2\alpha})} & \eta_u \geq \frac{1}{1+\mu^2} \;  (small \,\alpha)
	\end{cases},
	\end{align}
	where $\frac{1}{1+\mu^2}=\frac{1-(\cos{2n\alpha})^2(\cos{2\alpha})^2}{1-(\cos{2n\alpha})^2\cos{4\alpha}}$.
	
	In the large $\alpha$ regime ($\eta_u \leq \frac{1}{1+\mu^2}$) the optimal values of $\lambda_A, \lambda_B$ in Equation~(\ref{eq:fformAB}) and the form of the third POVM element read
	
	\begin{align}
	\label{eq:optla1}
	\lambda_A&=\frac{1}{2}(1+\tilde{c}_n(\tilde{c}-\tilde{s})) \quad \;
	\lambda_B=\frac{1}{2}(1-\tilde{c}_n(\tilde{c}+\tilde{s})) \\
	C&=I-A-B= \ket{\phi_C} \bra{\phi_C} \nonumber\\
	\label{eq:optla11}
	\ket{\phi_C}&=
	\begin{pmatrix}
	\sqrt{\nu_+}\\
	-\sqrt{\nu_-}\\
	\end{pmatrix}
	\end{align}
	
	which is in accordance with Equation~(\ref{eq:def_gate_M}).
	
	On the other hand, for $\eta_u \geq \frac{1}{1+\mu^2}$ the optimal solution turns into a projective measurement
	\begin{align}
	A&=\ket{\tilde{v}} \bra{\tilde{v}}\quad \;
	B=0 \\
	C&=I-A-B= \ket{\tilde{v}^\perp} \bra{\tilde{v}^\perp} \nonumber\\
	\ket{\tilde{v}}&=
	\begin{pmatrix}
	\sqrt{a}\\
	\sqrt{b}\\
	\end{pmatrix}
	\quad
	\ket{\tilde{v}^\perp}=
	\begin{pmatrix}
	\sqrt{b}\\
	-\sqrt{a}\\
	\end{pmatrix}
	\nonumber
	\end{align}   
	again in accordance with values of $a,b$ defined in Equation~(\ref{eq:def_gate_M}). 
	
	Here the POVMs are acting on the two-dimensional complex vector space on which matrices $A,B$ act. Elements of matrices $A,B$ define elements of the Choi operator $R_s$ and matrix $C$ serves as a guide for numerous possible completions of quantum operation $\mathcal{R}_s$ into a quantum channel. One can involve isometric dilation of such a channel and a suitable measurement on the ancillary system to realize quantum operation $\mathcal{R}_s$. Thus, in the regime of $\eta_u \leq \frac{1}{1+\mu^2}$ we introduce the following isometry $G$ from $\hilb{H}_0 \otimes \hilb{H}_1$ to $\hilb{H}_A \otimes \hilb{H}_2$
	
	\begin{align}
	\label{eq:defG}
	G=&\;\sqrt{\lambda_A}(\frac{c}{c_n}\ket{0}
	\bra{+}\otimes I + \frac{s}{s_n}\ket{0}
	\bra{-}\otimes \sigma_z) \nonumber\\
	&+\sqrt{\lambda_B}(\frac{s}{c_n}\ket{1}\bra{+}\otimes \sigma_z - \frac{c}{s_n}\ket{1} \bra{-}\otimes I) \nonumber\\
	&+\sqrt{\nu_+}\ket{2}
	\bra{+}\otimes I - \sqrt{\nu_-}\ket{2} \bra{-}\otimes \sigma_z)   ,
	\end{align}
	
	where $\hilb{H}_A$ is a three-dimensional Hilbert space (qutrit) spanned by orthonormal basis
	$\{\ket{0},\ket{1},\ket{2}\}$. 
	Using Eqs. (\ref{eq:2}),
	(\ref{eq:def_gate_M}) 
	we can verify that $G^\dagger G=I$.
	Let us calculate the conditional transformations if $\ket{0}$ or $\ket{1}$ are observed on the ancilla system. We obtain
	\begin{align}
	{}_A\bra{0}\otimes I_2\;G\;\ket{\psi_{n,0}}\otimes I_1&=
	\sqrt{\lambda_A}(c I+ i s \sigma_z)= \sqrt{\lambda_A} U_0 \nonumber\\
	{}_A\bra{0}\otimes I_2\; G \;\ket{\psi_{n,1}}\otimes I_1&=
	\sqrt{\lambda_A}(c I- i s \sigma_z)= \sqrt{\lambda_A} U_1 \nonumber     \\
	{}_A\bra{1}\otimes I_2\;G\;\ket{\psi_{n,0}}\otimes I_1&=
	\sqrt{\lambda_B}(s \sigma_z - i c I)= -i\sqrt{\lambda_B} U_0 \nonumber\\
	{}_A\bra{1}\otimes I_2\; G \;\ket{\psi_{n,1}}\otimes I_1&=
	i\sqrt{\lambda_B}(s \sigma_z + i c I)= i\sqrt{\lambda_B} U_1 \, .
	\end{align}
	We see that the desired unitary transformations were retrieved exactly, since the global phase is irrelevant.
	The overall success probability reaches its optimal value, since it reads
	\begin{align}
	P_{succ}&=  \frac{1}{2}\sum_{i=0,1} Tr_{A2}[G \;\ketbra{\psi_{n,i}}{\psi_{n,i}}\otimes \xi\; G^\dagger \;W_A\otimes I_2 ], \nonumber\\
	&= \lambda_A+\lambda_B = 1-\tilde{c}_n\tilde{s}
	\end{align}
	where $\xi$ is an arbitrary normalized input state to the retrieved transformation, and $W=\ketbra{0}{0} + \ketbra{1}{1}$.
	Similarly, in the regime of $\eta_u \geq \frac{1}{1+\mu^2}$ we define isometry 
	\begin{align}
	H=& a \ket{0} \bra{+}\otimes I + b \ket{0} \bra{-}\otimes \sigma_z  \nonumber\\
	&+b \ket{2}\bra{+}\otimes I - a \ket{2} \bra{-}\otimes \sigma_z 
	\end{align}
	from $\hilb{H}_0 \otimes \hilb{H}_1$ to $\hilb{H}_A \otimes \hilb{H}_2$, 
	where $a,b$ are defined in Equation~(\ref{eq:def_gate_M}). 
	One can check by direct calculation that $H^\dagger H = I$.
	The main difference with respect to Equation~(\ref{eq:defG}) is that in this regime one-dimensional subspace along vector $\ket{1}$ is not used, thus only outcomes in directions of $\ket{0}$ and $\ket{2}$ will appear and they correspond to success and failure, respectively. Indeed we obtain
	\begin{align}
	{}_A\bra{0}\otimes I_2\;H\;\ket{\psi_{n,0}}\otimes I_1&=
	\sqrt{q}(c I+ i s \sigma_z)= \sqrt{q} U_0 \nonumber\\
	{}_A\bra{0}\otimes I_2\; H \;\ket{\psi_{n,1}}\otimes I_1&=
	\sqrt{q}(c I- i s \sigma_z)= \sqrt{q} U_1 \nonumber     \\
	P_{succ}=
	q &= \frac{\tilde{s}^2_n}{2(1-\tilde{c} \tilde{c}_n)},    
	\end{align}
	which agrees with optimal value from Equation~(\ref{eq:fom_final}).
	
	Our next goal is to find an efficient quantum circuit realization of the optimal protocol.
	For that purpose we rewrite the isometries $G$, $H$ in the following form
	\begin{align}
	G=&(\ketbra{a_1}{+}+\ketbra{a_2}{-})\otimes \ketbra{0}{0} \nonumber\\
	&+(\ketbra{b_1}{+}+\ketbra{b_2}{-})\otimes \ketbra{1}{1}  \nonumber\\
	H=&(\ketbra{a'_1}{+}+\ketbra{a'_2}{-})\otimes \ketbra{0}{0} \nonumber\\
	&+(\ketbra{b'_1}{+}+\ketbra{b'_2}{-})\otimes \ketbra{1}{1}  
	\end{align}
	where

	\begin{align}
	F&=\ketbra{0}{0}-\ketbra{1}{1}+\ketbra{2}{2} \nonumber \\
	\ket{b_1}&=F\ket{a_1} \qquad \ket{b'_1}=F\ket{a'_1}= \ket{a'_1}\\
	\ket{b_2}&=-F\ket{a_2}  \qquad \ket{b'_2}=-F\ket{a'_2}=-\ket{a'_2}
	\end{align}
	and the vectors $\ket{a_1}, \ket{a_2}, \ket{a'_1}, \ket{a'_2}$ are as in Equation~(\ref{eq:def_gate_M}). 
	
	Now it is obvious we can write $G,H$ as
	\begin{align}
	\label{eq:circv1}
	G&=C_F \, M \otimes I\, C_X \\
	H&=M' \otimes I\, C_X 
	\end{align}
	where $C_X=I \otimes \ketbra{0}{0}+ \sigma_x \otimes \ketbra{1}{1}$ is a controlled-NOT (CNOT), $M=\ketbra{a_1}{+}+\ketbra{a_2}{-}$ and $M'=\ketbra{a'_1}{+}+\ketbra{a'_2}{-}$ are qubit to qutrit isometries and $C_F=I \otimes \ketbra{0}{0}+ F \otimes \ketbra{1}{1}$ is a controlled version of unitary $F$ applied to a qutrit. Finally, it is useful to realize what is the actual effect of $C_F$ if it is followed by the measurement in the basis $\{\ket{0},\ket{1},\ket{2}\}$ of the qutrit. Subspaces $\{\ket{00},\ket{01}\}$, $\{\ket{20},\ket{21}\}$ are left unchanged, while in subspace $\{\ket{10},\ket{11}\}$ we have $C_F\ket{10}=\ket{10}, C_F\ket{11}=-\ket{11}$, which is the same as action of $\sigma_z$ on the qubit. Thus, instead of implementing $C_F$ gate we might omit it 
	and only in case of measurement outcome $\ket{1}$ on the qutrit we will apply $\sigma_z$ on the unmeasured qubit to achieve the same resulting state as in the original circuit. In formula,
	\begin{align}
	\label{eq:circv2a}
	{}_A\bra{0}\otimes I_2\;G &={}_A\bra{0}M \otimes I\, C_X \\
	\label{eq:circv2b}
	{}_A\bra{1}\otimes I_2\; G &={}_A\bra{1}M \otimes \sigma_z\, C_X
	\end{align}
	In the regime of $\eta_u \geq \frac{1}{1+\mu^2}$ implementation of $C_F$ gate is not needed and the circuit differs only in the choice of qubit to qutrit isometry.
	Thus, we proved that the quantum circuit depicted in Figure \ref{fig:fig2} performs optimal storage and retrieval protocol for two unitary transformations of a qubit.

	\section{Optical implementation of the POVM}
	\label{app:isometry_par}
	
	\begin{figure}
		\includegraphics[width=\linewidth]{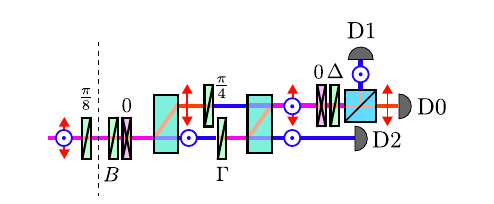}
		\caption{Optical implementation of the POVM. The circuit implements isometry $M$ followed by measurement in basis $|\uparrow H\rangle = |0_t\rangle$, $|\uparrow V\rangle = |1_t\rangle$, and $|\downarrow V\rangle = |2_t\rangle$. The first wave plate implements Hadamard operation, while the rest of the setup implements operator $K$. Polarization is denoted by arrow diagrams and colors, blue for vertical polarization and red for horizontal. The general superposition of both polarizations is indicated by the purple color. Half- and quarter-wave plates are depicted using green and purple color, respectively and the angular orientation of their fast axis with respect to the horizontal direction is described by black captions. POVM measurement is concluded by detecting a photon at one of three single-photon detectors D0-D2.}
		\label{fig:s1}
	\end{figure}

	Here we describe our implementation of the POVM measurement in detail. The corresponding experimental block is depicted in Figure~\ref{fig:s1}. It consists of half-wave plates (depicted in green) and quarter-wave plates (depicted in purple) which couple the horizontal $|H\rangle$, depicted in red, and vertical polarizations $|V\rangle$, depicted in blue. These two polarization modes initially occupy the same spatial mode. A calcite beam displacer separates the polarization modes spatially, displacing the horizontally polarized mode laterally by 6~mm, as depicted in the figure. The upper half-wave plate ensures that photons from the upper path always get to the detectors D0 and D1, while the bottom half-wave plate controls the probability of the photon in the bottom path reaching detector D2. The last stack of wave plates controls the splitting ratio between detectors D0 and D1. The black captions in the figure denote the angle between the fast axis of the wave plate and the horizontal direction. The subsequent click of single-photon detector D0, D1, or D2 heralds the detection of the respective POVM elements.

	The mapping between input modes $|H\rangle = |0\rangle$, $|V\rangle = |1\rangle$ and output modes  $|\uparrow H\rangle = |0_t\rangle$, $|\uparrow V\rangle = |1_t\rangle$, and $|\downarrow V\rangle = |2_t\rangle$ is described by the following relation:
	
	{\begin{widetext}
			\begin{equation}
			C =  \tilde{W}(\Delta, \pi) \tilde{W}(0, \frac{\pi}{2}) D_2 \left[|\uparrow\rangle\langle\uparrow| \otimes \sigma_x + |\downarrow\rangle\langle\downarrow| \otimes W(\Gamma, \pi)\right]   D_{1} W(0, \frac{\pi}{2}) W(B, \pi),
			\end{equation}
	\end{widetext}}
	where $D_1$ and $D_2$ describe the action of calcite crystals: 
	\begin{eqnarray}
	D_1 &=& |\uparrow H\rangle\langle H| + |\downarrow V\rangle\langle V|, \\
	D_2 &=& |\uparrow V\rangle\langle \uparrow V| + |\uparrow H\rangle\langle \downarrow H| + |\downarrow V\rangle\langle \downarrow V|.
	\end{eqnarray}
	The operator
	\begin{equation}
	W(x, y) = |\mathrm{L}_x\rangle\langle\mathrm{L}_x| + \exp(-iy)|\mathrm{L}_{x+\frac{\pi}{2}}\rangle\langle\mathrm{L}_{x+\frac{\pi}{2}}|
	\end{equation}
	describes the action of a rotated wave plate using its eigenstates $|\mathrm{L}_x\rangle = \cos(x)|H\rangle + \sin(x)|V\rangle$, $x$ being the orientation of the fast axis and $y$ the wave plate retardance. The operator $\tilde{W}(x, y)$ describes polarization coupling in the upper spatial mode, i.e. 
	\begin{equation}
	\tilde{W}(x,y) =
	\left(\begin{array}{ll}
	W(x,y) & 0 \\
	0 & 1
	\end{array}\right).
	\end{equation}
	
	Note that due to the optimization of the experiment, there effectively is Hadarmard gate $H$ at the output of our CNOT gate. Therefore, to find the wave plate angles $B$, $\Gamma$, and $\Delta$ for which the setup implements the isometry $M$ we solve the matrix equation 
	\begin{equation}
	C = M H^{\dagger},
	\end{equation}
	where $M = M_l$ for $|\alpha| \leq \alpha_t$ and $M = M_h$ otherwise. The transition point $\alpha_t$ is determined by the following equation:
	\begin{equation}
	\cos(2n\alpha_t) \cos(2\alpha_t - \pi/4) = \sqrt{2}/2.
	\end{equation}
	The operator $MH^{\dagger} = K_{l,h}$ reads
	\begin{equation}
	M_l H^{\dagger} = \left(\begin{array}{ll}
	k_{11} & k_{31} \\
	0 & 0 \\
	k_{31} & -k_{11} \\
	\end{array}\right),
	\end{equation}
	where
	\begin{align}
	k_{11} &= \sqrt{\frac{(1 + \cos 2\alpha)(1-\cos 2n\alpha)}{2(1-\cos 2\alpha \cos 2n\alpha)}},\\
	k_{31} &= \sqrt{\frac{(1 - \cos 2\alpha)(1+\cos 2n\alpha)}{2(1-\cos 2\alpha \cos 2n\alpha)}},
	\end{align}
	and	
	\begin{equation}
	M_h H^{\dagger} = \left(\begin{array}{ll}
	k_{11} & k_{12} \\
	k_{21} & k_{22} \\
	k_{31} & k_{32} \\
	\end{array}\right),
	\label{eq:K}
	\end{equation}
	where the elements are
	\begin{align}
	k_{11} & =  \sqrt{\lambda_a}\frac{\cos(\alpha)}{\cos(n\alpha)}, \label{eq:ks11} \\ 
	k_{12} & =  \sqrt{\lambda_a}\frac{\sin(\alpha)}{\sin(n\alpha)}, \label{eq:ks12} \\ 
	k_{21} & =  \sqrt{\lambda_b}\frac{\sin(\alpha)}{\cos(n\alpha)}, \label{eq:ks21} \\ 
	k_{22} & =  -\sqrt{\lambda_b}\frac{\cos(\alpha)}{\sin(n\alpha)}, \label{eq:ks22} \\ 
	k_{31} & =  \sqrt{\nu_p}, \label{eq:ks31} \\ 
	k_{32} & =  -\sqrt{\nu_p} \label{eq:ks32}. 
	\end{align}
	and 
	\begin{align}
	\lambda_a & =  \left(1 + \cos(2n\alpha) (\cos\alpha - \sin\alpha)\right)\frac{1}{2}, \\
	\lambda_b & =  \left(1 - \cos(2n\alpha) (\cos\alpha + \sin\alpha)\right)\frac{1}{2}, \\
	\nu_p & = \left( 1 - \cos^2 2\alpha + \sin 2\alpha \right)\frac{\cos 2n\alpha}{1 + \cos 2n\alpha}\\
	\nu_n & =  \left( -1 + \cos^2 2\alpha + \sin 2\alpha \right)\frac{\cos 2n\alpha}{1 - \cos 2n\alpha}.
	\end{align}
	
	A possible solution is 
	\begin{align}
	B & =  \frac{\pi}{4} + \frac{\beta}{4},\\
	\Gamma & =  \frac{\pi}{4} - \frac{\beta}{2},\\
	\Delta & =  \frac{\delta}{2},
	\end{align}
	where 
	\begin{align}
	\beta & =  \arctan\left(\frac{k_{32}}{k_{31}}\right), \\
	\gamma & =  \arcsin\left(\frac{-k_{31}}{\cos(\beta)}\right), \\
	\delta & =  \arcsin\left( -\frac{k_{22}\cos(\gamma) + k_{11}}{\sin\beta \sin^2 \gamma}  \right).
	\end{align}
	To implement isometry $M$, we apply a Hadamard gate on the input of this setup, i.e. $M = KH$. This Hadamard gate was already effectively present in the setup.
	
	This optical setup could also be reconfigured to perform unambiguous state discrimination between states
	$\left(|0\rangle + \exp(\pm2in\alpha)|1\rangle\right)/\sqrt{2}$ by setting wave plate angles to 
	\begin{align}
	B & =  \frac{\pi}{4},\\
	\Gamma & =  \frac{1}{2}\arcsin(\tan(|\alpha|)),\\
	\Delta & =  -\frac{\pi}{8},
	\end{align}
	and the first quarter-wave plate to 0 and the second to $\frac{\pi}{4}$.
	Although the quarter-wave plates are not necessary  for these operations, we included them in the experiment because the first two wave plates and the first calcite serve as a variable projector that we use for the tomographic characterization of the gate.

	\bibliographystyle{apsrev4-1}
	

\end{document}